\documentclass[10pt, reqno]{amsart}
\usepackage{amsmath}
\usepackage{amsfonts}
\usepackage{amssymb}
\usepackage{natbib}
\numberwithin{equation}{section}
\theoremstyle{plain}                % title and number in bold, text italic
\newtheorem{thm}{Theorem}[section]
\newtheorem{lem}[thm]{Lemma}
\newtheorem{prop}[thm]{Proposition}

\theoremstyle{definition}           % title and number in bold, text normal
\newtheorem{defn}[thm]{Definition}
\newtheorem{exam}[thm]{Example}

\newtheorem{ass}[thm]{Assumption}

\theoremstyle{remark}               % title and number in italic, text normal
\newtheorem{rem}{Remark}[section]

\newcommand{\sq}[2][n]{( #2 )_{#1\in\N}}
\newcommand{\sqb}[2][n]{\big( #2 \big)_{#1\in\N}}

\renewcommand\epsilon{\varepsilon}
 \DeclareMathOperator\cl{cl}

\DeclareMathOperator\sgn{sgn}

\newcommand{\R}{{\mathbb R}}
\newcommand{\N}{{\mathbb N}}
\newcommand{\PP}{{\mathbb P}}
\newcommand{\QQ}{{\mathbb Q}}
\newcommand{\EE}{{\mathbb E}}
\newcommand{\FF}{{\mathcal F}}

\newcommand{\BB}{{\mathcal B}}

\newcommand{\DD}{{\mathcal D}}
\newcommand{\MM}{{\mathcal M}}

\renewcommand{\AA}{{\mathcal A}}
\newcommand{\cd}{c\` adl\` ag }

\newcommand{\scl}[ 2]{\langle #1,#2 \rangle} % scalar product
\newcommand{\abs}[1]{\left| #1 \right|}
\newcommand{\set}[1]{\left\{#1\right\}} % curly brackets
\newcommand{\sets}[2]{\set{#1\,:\,#2}} % a set with "such that"
\newcommand{\inds}[1]{ {\mathbf 1}_{\set{#1}}} % indicator of a curly set
\newcommand{\ind}[1]{ {\mathbf 1}_{{#1}}} % indicator of a set
\newcommand{\seq}[1]{\set{#1_n}_{n\in\N}}

\newcommand{\eps}{\varepsilon}
\newcommand{\ld}{\lambda}

\DeclareMathOperator*{\Esssup}{esssup}
\newcommand{\el}{{\mathbb L}} %l-pees

\newcommand{\lone}{\el^1}

\newcommand{\linf}{\el^{\infty}}

\newcommand{\EN}{{\mathcal E}} % endowment process
\newcommand{\XX}{{\mathcal X}} % set of admissible wealths
\newcommand{\yq}{ Y^{\QQ}} % radon-nykodim of Q
\newcommand{\prf}[1]{ ( #1 )_{t\in [0,T]}} % process form - puts parenthesis

\newcommand{\hq}{\hat{\QQ}}
\newcommand{\hqy}{\hq^{y}}
\newcommand{\qn}{{{\QQ}^{(n)}}}
\newcommand{\qnn}{\{\qn\}_{n\in\N}}
\newcommand{\yqn}{Y^{\qn}}

\newcommand{\norm}[1]{{||#1||}}

\newcommand{\slzer}{{\mathcal V}}

\newcommand{\slzm}{{\mathcal V}^{\MM}}

\newcommand{\slzmk}{\slzm_{\kappa}}
\newcommand{\slzmkp}{\slzm_{\kappa+}}

\newcommand{\eitk}[1]{\EE\int_0^T #1\, d\kappa_t}
\newcommand{\eqitk}[1]{\EE^{\QQ}\int_0^T #1\, d\kappa_t}
\newcommand{\normm}[1]{\| #1 \|_{\MM}}
\newcommand{\tT}{{\tau}}

\newcommand{\hcx}{\hat{c}^{x}}
\newcommand{\fU}{{\mathbf U}}
\newcommand{\fV}{{\mathbf V}}
\newcommand{\fVE}{\fV^{\EN}}
\newcommand{\fI}{{\mathbf I}}

\newcommand{\OO}{{\mathcal O}}
\DeclareMathOperator{\Dom}{Dom}

\DeclareMathOperator{\ba}{{\mathbf{ba}}}
\DeclareMathOperator{\Sing}{\mathrm{Sing}}
\newcommand{\bam}{\ba^{\MM}}

\newcommand{\PPk}{\PP_{\kappa}}
\newcommand{\QQk}{\QQ_{\kappa}}
\newcommand{\MMk}{\MM_{\kappa}}
\newcommand{\DDk}{\DD_{\kappa}}
\newcommand{\LL}{{\mathcal L}}
\newcommand{\UU}{{\mathcal U}}

\newcommand{\prfi}[1]{(#1)_{t\in [0,\infty)}}
\newcommand{\sequ}[1]{\{#1^{(n)}\}_{n\in\N}}

\newcommand{\nux}{\nu^{x}}
\newcommand{\Hx}{\pi ^x}
\renewcommand{\cl}[1]{{\mathrm{Cl}}\left(#1\right)}

\begin{document}

\title[Stochastic Clock]{Utility Maximization with a Stochastic Clock and an Unbounded
Random Endowment}
\author{Gordan \v Zitkovi\' c }
\thanks{{\bf Acknowledgements.} The author would like to thank the
anonymous referees for a number of useful suggestions and
improvements. I am also indebted to the probability seminar
participants at Carnegie Mellon University, Brown University and
Boston University. \newline \hskip1cm During the creation of this
text, the author was partially supported by the National Science
Foundation under grant DMS-0139911. Any opinions, findings, and
conclusions or recommendations expressed in this material are
those of the author and not necessarily reflect the views of the
National Science Foundation.}
\address{
    Gordan \v Zitkovi\' c,
    Department of Mathematical Sciences,
    7209 Wean Hall,
    Carnegie Mellon University, Pittsburgh, PA 15217, USA}
\email{zitkovic@cmu.edu}
\urladdr{www.andrew.cmu.edu/$\sim$gordanz}
\date{\today}
\begin{abstract}
We introduce a linear space of finitely additive measures to treat
the problem of optimal expected utility from consumption under a
stochastic clock and an unbounded random endowment process. In
this way we establish existence and uniqueness for a large class
of utility maximization problems including the classical ones of
terminal wealth or consumption, as well as the problems depending
on a random time-horizon or multiple consumption instances. As an
example we treat explicitly the problem of maximizing the
logarithmic utility of a consumption stream, where the local time
of an Ornstein-Uhlenbeck process acts as a stochastic clock.
\end{abstract}
\maketitle
% Introduction

\section{Introduction}

When we speak of the expected utility, we  usually have one of the
following two cases in mind: expected utility of consumption on a
finite interval, or the expected utility of terminal wealth at
some future time point. These two cases correspond to the two of
the historically most important problem formulations in the
classical calculus of variations and optimal (stochastic) control
- the {\em Meyer formulation} $\EE[\int_0^T L(s,x(s))\, dt]\to
\max$ and the {\em Lagrange formulation} $\EE[\psi(x(T))]\to\max$,
where $x(\cdot)$ denotes the controlled state function or
stochastic process, and $L$ and $\psi$ correspond to the
optimization criteria. These formulations owe a great deal of
popularity to their analytical tractability; they fit very well
into the framework of the dynamic programming principle often used
to tackle optimal control problems. Even though there is a number
of problem formulations in the stochastic control literature  that
cannot be reduced to either a Meyer or a Lagrange form (see
Section 2.7, pages 85-92 of \citet{YonZho99}, for an overview of
several other classes of stochastic control models), the expected
utility theory in contemporary mathematical finance seems to lag
behind in this respect. The introduction of convex duality into
the treatment of utility maximization problems by
\cite{KarLehShr87} and \cite{KarLehShrXu91}, as well as its
further development in \cite{KraSch99}, \cite{CviSchWan01},
\cite{KarZit03} and \cite{HugKra02} (to list but a small subset of
the existing literature) give hope that this lag can be overcome.

This paper aims at formulating and solving a class of utility
maximization problems of the {\em stochastic clock} - type in
general incomplete semimartingale market with locally bounded
stock prices and a possibly unbounded random endowment
process. More specifically, our objective is to provide a mathematical
framework for  maximizing  functionals
of the form $\EE[\int_0^T U(\omega, t, c_t)\, d\kappa_t]$, where $U$
is a time- and uncertainty-dependent utility function (a utility
random field), $c_t$ is the consumption density process, and
$\kappa_t$ is an arbitrary non-decreasing right-continuous adapted
process on $[0,T]$ with $\kappa_T=1$. Two particular choices $\kappa_t= t/T$, and
$\kappa_t=\inds{t= T}$ correspond to the familiar Meyer and
Lagrange formulations of the utility maximization problem, but
there are many other financially feasible ones. The problems of
maximization of the expected utility at terminal time $T$, when
$T$ is a stopping time denoting the retirement time or a default time, form
a class of examples. Another class consists of problems with the compound expected
utility sampled at a sequence of stopping times.
Furthermore, one could model random consumption prohibition by
setting $\kappa_t=\int_0^t \inds{R_u\in C}\, du$ for some index process
$R_t$ and a set $C\subseteq \R$.

The notion of a stochastic clock has already been explicitly
present in \citet{GolKal03} (where the phrase {\em stochastic
clock} has been introduced), and implicitly in \citet{Zit99},
\citet{Zit02} and \citet{KarZit03}. \citet{GolKal03} treat the
case of a logarithmic utility with no random endowment process,
under additional assumptions on existence of the optimal dual
process. \citet{KarZit03} establish existence and uniqueness of
optimal consumption process in an incomplete semimartingale market
in the presence of a bounded random endowment. Their version of
the stochastic clock is, however, relatively limited - it is
required to be a deterministic process with no jumps on $[0,T)$.
This assumption was crucial for their treatment of the problem
using convex duality, and is related to the existence of a \cd
version of the optimal dual process. Related to the notion of a
stochastic clock is the work \cite{BlaElkJeaMar03}, which deals
with the utility maximization on a random horizon not necessarily
given by a stopping time. Also, recent work of \cite{BouPha03}
treats the wealth-path dependent utility maximization. The authors
use a duality relation between the wealth processes and a suitably
chosen class of dual processes viewed as optional measures on the
product space $[0,T]\times\Omega$.

In the present paper we extend the existing literature in several
ways. We prove existence and describe the structure of the optimal
strategy under fairly unrestrictive assumptions on the financial
market and the random endowment process.

First, we allow for a general stochastic clock and a general
utility satisfying the appropriate version of the requirement of
 reasonable elasticity of \citet{KraSch99}.

 Second, we allow a random endowment process that
is not necessarily bounded, we only require a finite upper-hedging
price for the total endowment at time $t=T$. The case of a
non-bounded random endowment in the utility maximization
literature has been considered in \citet{HugKra02}, but only in
the case of the utility of terminal wealth, and using techniques
different from ours. The only restriction warranting discussion is
the one we place on the jumps of the stock-price process $S$.
Namely, we require $S$ to be locally bounded. The reason for this
requirement (not present in \cite{KarZit03}, but appearing in
\cite{HugKra02}) is that the random endowment process is not
assumed to be bounded anymore, and the related notion of
acceptability (developed only in the locally-bounded setting) has
to be employed.

Finally, we present an example in which we deal completely
explicitly with a utility maximization problem in an It\^
o-process market model with constant coefficients where the
stochastic clock is the local time at $0$ of an Ornstein-Uhlenbeck
process. This example illustrates how the uncertainties in the
future consumption prohibitions introduce the incompleteness into
the market, and describes the optimal strategy to face them.

In order to tackle the problem of utility maximization with the
stochastic clock we cannot depend on existing techniques. We still
use the convex-duality approach, but in order to formulate and
solve the dual problem we introduce and study the properties of
two new Banach spaces - one of consumption densities and the other of
finitely-additive measures. Also, we
simplify the formulation of the standard components of the
convex-duality treatment by defining the dual objective function
directly as the convex conjugate of the primal objective function
in the suitably coupled pair of Banach spaces. In this way, the
 mysterious regular parts of the finitely-additive
counterparts of the martingale measures used in \cite{CviSchWan01}
and \cite{KarZit03} in the definition of the dual problem, appear
in our treatment more naturally, in an a posteriori fashion.

The paper is organized as follows. After this Introduction, Section 2.~describes the model of the financial
market and  poses the utility maximization problem. In Section 3.~we introduce the functional-analytic setup
needed for the convex-duality treatment of our optimization problem. Section 4.~introduces the
convex conjugate of the utility functional and states the main result. An example admitting an explicit
solution is treated in Section 5. Finally, Appendix A contains the proof of our main result.

%Section 2

\section{The Financial Market and the Optimization Problem}

\subsection{The Stock-price Process} We consider a financial market on a finite horizon $[0,T]$, $T\in(0,\infty)$,
consisting of a $d$-dimensional locally bounded semimartingale
$\prf{S_t}=\prf{S^1_t,\ldots,S^d_t}$. The process $\prf{S_t}$ is
defined on a stochastic
base $(\Omega,\FF,$ $\prf{\FF_t},\PP)$ satisfying the {\em
usual conditions}. For simplicity we also assume that $\FF_0$ is
$\PP$-trivial and that $\FF=\FF_T$. Together with the stock-price process $\prf{S_t}$, there is
a num\' eraire asset $S^0$, and all values will be denominated in terms of $S^0_t$. This amounts
to the standard assumption that $\prf{S^0_t}$ is equal to the constant process $1$.

\subsection{Admissible Portfolio Processes} A financial agent invests in the market according to an $\prf{\FF_t}$-predictable $S$-integrable
$d$-dimensional {\bf portfolio process} $\prf{H_t}$. The
stochastic integral $\prf{(H\cdot S)_t}$ is called the {\bf gains
process} and represents the net gains from trade for the agent who
holds a portfolio with $H^k_t$ shares of the asset $k$ at time
$t$, for $k=1,\ldots, d$.

 A portfolio process $\prf{H_t}$ is called {\bf admissible} if there exists a constant $x\in\R$
such that $x+(H\cdot S)_t\geq 0$ for all $t\in [0,T]$, with
probability 1. Furthermore, an admissible process $\prf{H}$ is called
{\bf maximal admissible} if there exists no other admissible
process $\prf{\tilde{H}}$ such that
\[ (H\cdot S)_T\leq (\tilde{H}\cdot S)_T\ \text{a.s., and }\
\PP[(H\cdot S)_T< (\tilde{H}\cdot S)_T]>0.\] The family of all
processes $\prf{X^H_t}$ of the from $X^H_t\triangleq (H\cdot S)_t$, for an admissible $H$, will be denoted by $\XX$. The class of
processes $\prf{X^H_t}\in\XX$ corresponding to maximal admissible
portfolio processes $\prf{H}$, will be denoted by $\XX_{\max}$.

We complement the wide-spread notion of admissibility by the
less-known notion of acceptability (introduced in
\citet{DelSch97a}) because admissibility is not adequate for
dealing with non-bounded random endowment processes, as it has
been shown in the context of utility maximization from terminal
wealth in \citet{HugKra02}. A portfolio process $\prf{H}$ is
called {\bf acceptable} if it admits a decomposition $H=H^+-H^-$
with $H^+$ admissible and $H^-$ maximal admissible.

\subsection{Absence of Arbitrage} In order to rule out the arbitrage
opportunities in our market, we state the following assumption
\begin{ass} There exists a probability \label{ass:NFLVR}
measure $\QQ$ on $\FF$, equivalent to $\PP$, such that the process
$\prf{S_t}$ is a $\QQ$-local martingale.
\end{ass}
 It has been shown in the celebrated paper of \citet{DelSch94}, that
 the condition in Assumption \ref{ass:NFLVR} is equivalent to the notion of
 No Free Lunch With Vanishing Risk (NFLVR) -  a concept closely related to, and
 only slightly stronger than the
 classical notion of absence of arbitrage.
 The condition NFLVR is
 therefore widely excepted as an operational proxy for the absence
 of arbitrage, and the Assumption \ref{ass:NFLVR} will be in force
 throughout the rest of the paper.

The set of all measures $\QQ\sim\PP$ as in Assumption
\ref{ass:NFLVR} will be denoted by $\MM$, and we will refer to the
elements of $\MM$ as the {\bf equivalent local martingale
measures}.

\subsection{Endowment and Consumption} Apart from being allowed to invest in the market in an admissible way, the agent
\begin{itemize}
\item[(a)] is continuously getting funds from an exogenous source
(random endowment), and
\item[(b)] is allowed to consume parts of
his wealth as the time progresses.
\end{itemize}
 These capital
in- and out-flows are modelled by  non-decreasing processes
$\prf{\EN_t}$ and $\prf{C_t}$ in $\slzer$, where $\slzer$ denotes
the set of all \cd $\prf{\FF_t}$-optional processes  vanishing at
$0$ whose paths are of finite variation. Here, and in the rest of the paper,
we always identify $\PP$-indistinguishable processes without explicit mention.

The linear space $\slzer$ can be given a structure of a vector lattice,
by equipping it  with a partial order $\preceq$, compatible with
its linear structure: we declare
\[ F^1\preceq F^2 \ \text{if the process}\
\prf{F^2_t-F^1_t}\  \text{has non-decreasing paths.}\] The cone of
all non-decreasing processes in $\slzer$ is the {\bf positive
cone} of the vector lattice $\slzer$ and we denote it by
$\slzer_+$. Also, the {\bf total variation} process
$\prf{\abs{F}{}_t}\in\slzer_+$ is associated with each $F\in\slzer$.

The process introduced in (a) above and denoted by $\prf{\EN_t}\in\slzer_+$ represents the {\bf random endowment}, i.e.
the value $\EN_t$ at time $t\in [0,T]$ stands for the
cumulative amount of endowment received by the agent during the
interval $[0,t]$. The process $\prf{\EN_t}$ is given exogenously, and we assume that the
agent exerts no control over it. On the other hand, the amount and distribution of the consumption
is decided by the agent, and  we model the agent's consumption strategy by
the {\bf consumption process} $\prf{C_t}\in\slzer_+$; the value $C_t$ is
the cumulative amount spent on consumption throughout the interval $[0,t]$.
We will find it useful in the later sections to interpret the processes in $\slzer_+$ as
optional random measures on the Borel sets of $[0,T]$.

\subsection{Wealth Dynamics} Starting from the initial wealth of $x\in\R$ (which can be negative) and the endowment process $\prf{\EN_t}$, our agent is free to choose an acceptable portfolio process $\prf{H_t}$
and a consumption process $\prf{C_t}\in\slzer_+$. These two
processes play the role of  the  controls of the system. The
resulting {\bf wealth process} $\prf{X^{(x,H,C)}_t}$ is given by
the {\bf wealth dynamics equation}
\begin{equation}
    %\nonumber
    \label{equ:wealthprocess}
       X^{(x,H,C)}_t\triangleq x+(H\cdot S)_t-C_t+\EN_t, \quad t\in [0,T].
\end{equation}
A consumption process $\prf{C}\in\slzer_+$ is said to be
$(x,\EN)$-{\bf financeable} if there exists an acceptable
portfolio process $\prf{H}$ such that $X^{(x,H,C)}_T\geq 0$ a.s.
The class of all $(x,\EN)$-financeable consumption processes will
be denoted by $\AA(x,\EN)$, or simply by $\AA(x)$, when there is
no possibility of confusion.
\begin{rem}The introduction of the concept of financeability
which suppresses the explicit mention of the portfolio process
$\prf{H_t}$, will be justified later when we specify the objective
(utility) function. It will depend only on the consumption and not
on the particular portfolio process used to finance it, so we will
find it useful to formulate a {\em static} version of the
optimization problem in which the portfolio process $\prf{H_t}$ will not appear at all.
\end{rem}

\begin{rem}The notion of financeability imposes a weak solvency restriction on the amount of
wealth the agent can consume: even  though the total wealth
process $\prf{X^{(x,H,C)}_t}$ is allowed to take strictly negative
values before the time $T$, the agent must plan the consumption
and investment in such a way to be able to pay all the debts by
the end of the planning horizon with certainty.  In other words,
borrowing is permitted, but only against the future endowment so
that there is no chance of default. With this interpretation it
makes sense to allow the initial wealth $x$ to take negative
values - the initial debt might very well be covered from the
future endowment. Finally, we stress that our notion of
financeability differs from the one introduced in
\citet{ElkJea98}, where no borrowing is allowed. A treatment of a
consumption problem with such a stringent financeability condition
seems to require a set of techniques different from ours and we
leave it for future research.
\end{rem}
\subsection{A Characterization of Financeable Consumption Processes}
 In the treatment
of our utility-maximization problem in the main body of this
paper, the so-called {\bf budget-constraint}-characterization
of the set $\AA(x)$ will prove to be useful.
The idea is to describe the financeable consumption processes in terms of a
set of linear inequalities.
We provide such a characterization it in the following
proposition under the assumption that the random variable $\EN_T$ (denoting the total cumulative endowment
over the horizon $[0,T]$)
{\bf admits an upper-hedging price}, i.e.
$\UU(\EN_T)\triangleq \sup_{\QQ\in\MM} \EE^{\QQ}[\EN_T]<\infty$.
\begin{prop}
    \ \label{pro:charadm}
    Suppose that the total endowment $\EN_T$ admits an upper-hedging price, i.e. $\UU(\EN_T)<\infty$.
    Then, the process $\prf{C_t}\in\slzer_+$ is $(x,\EN)$-financeable if and only if
\begin{equation}
    %\nonumber
    \label{equ:charadm}
    \begin{split}
        \EE^{\QQ}[C_T]\leq x+\EE^{\QQ}[\EN_T],\ \forall\,\QQ\in\MM.
    \end{split}
\end{equation}
\end{prop}
\begin{proof}
\label{prf:charadm}
   {\em `` only if '':} Assume first that $\prf{C_t}\in\AA(x,\EN)$, and pick an acceptable  portfolio process
   $\prf{H_t}$ such that the wealth process
    $\prf{X^{(x,H,C)}_t}$ defined in (\ref{equ:wealthprocess}) satisfies $X^{(x,H,C)}_T\geq 0$ a.s.
    By the definition of acceptability, there exists a
    decomposition $H=H_+-H_-$ into an admissible $H_+$ and a maximal admissible
     $H_-$ portfolio processes. Let $\MM'$ be the set of  all
    $\QQ\in\MM$ such that $\prf{(H^-\cdot S)_t}$ is a $\QQ$-uniformly integrable
    martingale. For any $\QQ\in\MM$ the process
    $\prf{(H^+\cdot S)_t}$ is a $\QQ$-local martingale bounded from
    below, and therefore a $\QQ$-supermartingale. Hence,
    $\prf{(H\cdot S)_t}$ is a $\QQ$-supermartingale for all
    $\QQ\in\MM'$ and
\begin{equation}
    %\nonumber
    \label{equ:alm}
    \begin{split}
        0&\leq \EE^{\QQ}[X^{(x,H,C)}_T|\FF_0]=x+\EE^{\QQ}[(H\cdot
        S)_T|\FF_0]+\EE^{\QQ}[\EN_T-C_T|\FF_0]\\ &\leq
        x+\EE^{\QQ}[\EN_T]-\EE^{\QQ}[C_T],\ \text{for all
        $\QQ\in\MM'$.}
    \end{split}
\end{equation}
      The set $\MM'$ of all $\QQ\in\MM$ such that $H^{-}\cdot S$ is a $\QQ$-uniformly integrable martingale is
    convex and dense in $\MM$ in the total variation norm (see \citet{DelSch97a}, Theorem 5.2).
    Therefore, the claim follows from (\ref{equ:alm}) and the
    density of $\MM'$ in $\MM$.

    {\em ``if'':} Let $\prf{C_t}\in\slzer_+$ be a process satisfying $\EE_{\QQ}[C_T]\leq x+\EE_{\QQ}[\EN_T]$ for
    all $\QQ\in\MM$.
    Since $\EN_T\geq 0 $ admits an upper-hedging price, there exists a
    constant $p>0$ and a maximal
    admissible portfolio process $\prf{H^{\EN}_t}$ such that
    $p+(H^{\EN}\cdot S)_T\geq \EN_T$ a.s. (see Lemma 5.13 in \citet{DelSch98}).
     Define the process
\begin{equation}
    \nonumber
    \label{equ:ifproof}
    \begin{split}
        F_t\triangleq \Esssup_{\QQ\in\MM}\,
    \EE_{\QQ}[C_T-\EN_T+p+(H^{\EN}\cdot S)_T|\FF_t],
    \end{split}
\end{equation}
    and note that $F_0\leq x+p$.
     $\prf{F_t}$ is a nonnegative
     $\QQ$-supermartingale for all $\QQ\in\MM$,
     permitting a \cd modification (see \citet{Kra96}, Theorem 3.2), and thus
     the Optional Decomposition Theorem (see \citet{Kra96}, Theorem 2.1) asserts the existence of an admissible portfolio
     processes $\prf{H^F_t}$ and a finite-variation process $\prf{G_t}\in\slzer_+$ such that
\begin{equation}
    \nonumber
    \label{equ:ifproof2}
    \begin{split}
        F_t=F_0+(H^F\cdot S)_t-G_t, \ \text{for all $t\in [0,T]$, a.s.}
    \end{split}
\end{equation}
        If follows that $x+p+(H^F\cdot S)_T\geq C_T-\EN_T+p+(H^{\EN}\cdot S)_T$, so
        for the acceptable portfolio process $\prf{H_t}$, defined by $H_t\triangleq H^F_t-H^{\EN}_t$ we have
         $x+ (H\cdot S)_T-C_T+\EN_T\geq 0$.
\end{proof}
\subsection{The Utility Functional and the Primal Problem}
In order to define the objective function of our optimization
problem, we need two principal ingredients: a utility random field
and the stochastic clock process.

The notion  of a utility random field as defined below has
appeared in \citet{Zit99} and \citet{KarZit03}, and we use it
because of its flexibility and good analytic properties - there
are no continuity requirements in the temporal argument, and so it
is well suited for our setting.

As for the notion of a stochastic clock, it models the the agent's
(either endogenously or exogenously imposed) notion of passage of time with respect to which the consumption
rate is being calculated and  utility accumulated. Several examples often appearing
in mathematical finance will be given below. Before that let us  give the
formal definition of the concepts involved:
\begin{defn}\ \label{def:utilfield}\nopagebreak
\begin{enumerate}\nopagebreak
\item A {\bf utility random field} $U:\Omega\times [0,T]\times
(0,\infty)\to\R$ is an $\FF\otimes\BB[0,t]\otimes \BB(0,\infty)$ -
measurable function satisfying the following conditions.
\begin{enumerate}
\item \label{U1} For a fixed $(\omega,t)\in\Omega\times [0,T]$,
the function $x\mapsto U(\omega, t,x)$ is a utility function, i.e.
a strictly concave, increasing $C^1$-function satisfying
the Inada conditions:\[  \lim_{x\to 0+} U_x(\omega,
t,x)=\infty\ \text{ and }\ \lim_{x\to\infty} U_x(\omega,
t,x)=0, \text{ a.s, }\]
where $U_x(\cdot, \cdot, \cdot)$ denotes the derivative with respect to the last argument.
\item \label{U2} There are continuous,
strictly decreasing (non-random) functions $K_i:(0,\infty)\to
(0,\infty)$, $i=1,2$ satisfying $\limsup_{x\to\infty}
\frac{K_2(x)}{K_1(x)} <\infty$, and constants $G<D\in\R$ such that
 we have
\[ K_1(x)\leq U_x(\omega,t,x) \leq K_2(x),\]
for all $(\omega, t,x)\in \Omega\times [0,T]\times (0,\infty)$,
 and
\[ G\leq U(\omega, t, 1)\leq D,\]
for all $(t,\omega)\in [0,T]\times\Omega$.
 \item \label{U4} For every
optional process $\prf{c_t}$,  the process $\prf{U(\omega, t,c_t)}$ is
optional.
\item $U$ is {\bf reasonably elastic}, i.e. it satisfies ${\mathrm{AE}}[U]<1$, where
${\mathrm{AE}}[U]$ denotes the {\bf asymptotic elasticity} of the random field $U$, defined by
\[{\mathrm{AE}}[U]\triangleq \limsup_{x\to\infty}
\left(\Esssup_{(t,\omega)\in [0,T]\times\Omega} \frac{x
U_x(\omega, t,x)}{U(\omega, t,x)}\right).\]
\end{enumerate}
\item The {\bf stochastic clock } $\prf{\kappa_t}$ is an arbitrary
process in  $\slzer_+$, such that $\kappa_T=1$, a.s.
\end{enumerate}
\end{defn}
\begin{rem}
The requirement $\kappa_T=1$ in the definition above is a mere
normalization. We impose it in order to be able to work with
probability measures on the product space $[0,T]\times\Omega$ (see
Section \ref{sec:funanal}.)
\end{rem}
We are now in the position to define the notion of a {\bf utility
functional} which takes consumption processes as arguments and
returns their expected utility. This expected utility (as defined
below in \ref{equ:udef}) will depend only on the part of the
consumption process $\prf{C_t}$ admitting a density with respect
to the stochastic measure $d\kappa$, so that the choice of a
consumption plan with a nontrivial component singular to $d\kappa$
would be clearly suboptimal. For that reason we restrict our
attention only to consumption processes $\prf{C_t}$ whose
trajectories are absolutely continuous with respect to $d\kappa$,
i.e. only processes of the form $C_t=\int_0^t c_t\,d\kappa_t$, for
a nonnegative optional process $\prf{c_t}$ which we will refer to
as the {\bf consumption density} of the consumption process
$\prf{C_t}$. For simplicity, we shall assume that the random
endowment admits a $d\kappa$-density $\prf{e_t}$ in that
$\EN_t=\int_0^t e_u\, d\kappa_u$, for all $t\in [0,T]$, a.s. This
assumption is clearly not necessary since the restrictions, which
the size of the random endowment places on the choice of the
consumption process, depend only on the value $\EN_T$, as we have
shown in Proposition \ref{pro:charadm}. We impose it in order to
simplify notation by having all ingredients defined as elements of
the same Banach space (see Section \ref{sec:funanal}.)

 The utility derived from a consumption
process should therefore be viewed as a function of the
consumption density $\prf{c_t}$ and we define the {\bf utility
functional} as a function on the set of optional processes:
\begin{equation}
    %\nonumber
    \label{equ:udef}
    \begin{split}
        \fU(c)\triangleq \EE\int_0^T U(\omega, t,
        c_t)\,d\kappa_t,\ \text{for an optional process
        $\prf{c_t}$.}
    \end{split}
\end{equation}  To deal
with the possibility of ambiguities of the from $(+\infty)- (-\infty)$
in the definition above, we adopt the following convention,
standard in the utility-maximization literature:  when the
integral $\EE\int_0^T \big(U(\omega, t, c_t)\big)^- \,d\kappa_t$
of the negative part $\big(U(\omega, t, c_t)\big)^-$ of the
integrand from (\ref{equ:udef}) takes the value $-\infty$, we set
$\fU(c)=-\infty$. In other words, our financial agent is not
inclined towards the risks that defy classification, as far as the
utility random field $U$ is concerned.
Finally, we add a mild technical integrability assumption on the utility functional $U$.
It is easily
seen to be satisfied by all our examples, and it is crucial for the simplicity of the proof of Proposition
\ref{pro:propfv}.
\begin{ass}\label{ass:delta}
For any nonnegative optional process $\prf{c_t}$ such that $\fU(c)>-\infty$
and any constant $0<\delta<1$ we have $\fU(\delta c)>-\infty$
\end{ass}

\subsection{Examples of Utility Functionals}
\begin{exam}[Utility Random Fields]\
\label{exa:I}\begin{enumerate}
\item Let $U(x)$ be a utility function satisfying
$\limsup_{x\to\infty} \frac{x U'(x)}{U(x)}<1$. Also, suppose there
exist functions $A:(0,\infty)\to \R$ and $B:(0,\infty)\to
(0,\infty)$ such that $U(\delta x)>A(\delta)+B(\delta) U(x)$, for
all $\delta>0$ and $x>0$. A family of examples of such utility functions is
supplied by the HARA family
\[ U_{\gamma}(x)=
\begin{cases} \frac{x^{\gamma}-1}{\gamma},& \gamma<1,
\gamma\not=0, \\
\log(x) & \gamma=0,
\end{cases}\] Then, the
(deterministic) utility random field
\[ U(\omega, t, x)=\exp(-\beta t) U_{\gamma}(x) \]
conforms to Definition \ref{def:utilfield}, and satisfies
Assumption \ref{ass:delta}.
\item If we take a finite number $n$ of $\prf{\FF_t}$-stopping times $\tau_1,\dots, \tau_n$,
positive constants
$\beta_1,\dots, \beta_n$ and $n$ utility functions $U^1(\cdot),\dots, U^n(\cdot)$
as in (1) and define
\[ U(\omega, t, x)=\sum_{i=1}^{n} \exp(-\beta_i t) U^i(x)\inds{t=\tau_i(\omega)},\]
the random field $U$ can be easily redefined on the complement of the union of the graphs
of stopping times $\tau_i$, $i=1,\dots, n$ to yield a utility random field
satisfying Assumption \ref{ass:delta}.
\end{enumerate}
\end{exam}
\begin{exam}[Stochastic clocks I]\
\begin{enumerate}
\item Set $\kappa_t=t$, for $t\leq T=1$. The utility functional
takes the from of {\bf utility of consumption} $\fU(c)=\EE\int_0^1 U(\omega,
t, c_t)\, dt$.
\item For $\kappa_t=0$ for $t<T$, and $\kappa_T=1$, we are looking at the {\bf utility
of terminal wealth $\EE[U(X_T)]$}, where $U(x)=U(\omega, T, x)$. Formally, we would get
an expression of the form $\fU(c)=\EE[U(\omega, T, c_T)]$, but clearly $c_T=X_T$ in all but suboptimal cases.
\item A combination $\kappa_t=t/2$ for $t<T=1$, and $\kappa_T=1$, of the two cases above models
the {\bf utility of consumption and terminal wealth} $\fU(c)=\EE[\int_0^1 U(\omega, t, c_t)\, dt+U(X_T)]$.
\end{enumerate}
\end{exam}
\bigskip
\begin{exam}[Stochastic clocks II]\
\begin{enumerate}
\item Let $\tau$ be an a.s. finite $\prf{\FF_t}$-stopping time.
We can think of $\tau$ as a random horizon such as the retirement
time, or some other market-exit time. Then the stochastic clock
$\kappa_t=0$, for $t<\tau$, and $\kappa_t = 1$ for $t\geq \tau$,
models the {\bf expected utility} $\EE[U(X_{\tau})]$ {\bf of the
wealth at a random time $\tau$}. The random endowment $\EN_{\tau}$
has the interpretation of the retirement package. In the case in
which the random horizon $\tau$ is unbounded, it will be enough to
apply a deterministic time-change to fall back within the reach of
our framework.
\begin{rem}
As the anonymous referee points out, the case of a random horizon
$\tau$ given by a mere random (as opposed to a stopping) time can
be included in this framework by defining $\kappa$ as the
conditional distribution of $\tau$, given the filtration
$\prf{\FF_t}$, as in \cite{BlaElkJeaMar03}.
\end{rem}
\item The example in (1) can be extended to go well with the utility function from Example \ref{exa:I} (2).
For an $n$-tuple of $\prf{\FF_t}$-stopping times, we set
\[ \kappa_t=\sum_{i=1}^n \frac{1}{n} \inds{t\geq \tau_i},\]
so that
\[ \fU(c)=\frac{1}{n} \sum_{i=1}^n  \EE[ \exp(-\beta_i \tau_i) U^i(c_{\tau_i})].\]
\item if we set $\kappa_t=1-\exp(-\beta t)$ for $t<\tau$ and $\kappa_t=1$, for $t\geq \tau$, we can add
consumption to the example in (1)  \[\fU(c)=\EE[\int_0^{\tau} \exp({-\beta t}) U(\omega, t, c_t)\, dt+
(1-\exp(-\beta \tau)) U(X_{\tau})],\] modelling the {\bf utility from consumption up to- and the
remaining wealth at the random time $\tau$.} The possibly inconvenient  factor $(1-\exp(-\beta \tau))$ in front of
the terminal utility term can be dealt away with by absorbing it into the utility random field.
\end{enumerate}
\end{exam}
\enlargethispage*{20pt}
\begin{exam}[Stochastic clocks, IV]\
\begin{enumerate}
\item In this example we model the situation when the agent is allowed to withdraw the consumption funds
only when a certain index process $R_t$ satisfies $R_t\in C$, for
some Borel set $C\subseteq \R$. In terms of the stochastic clock
$\kappa$, we have $\kappa_t=\min(\int_0^t \inds{R_t\in C}\,
dt,1)$. The $R_t$ could take a role of a political indicator in an
unstable economy where the individual's funds are under strict
control of the government. Only in periods of political stability, i.e. when $R_t\in C$,
 are the {\bf withdrawal constraints} relaxed and we are
allowed to withdraw funds from the bank. It should be stressed
here that the time horizon in this example is not deterministic.
It is given by the stopping time \[ \inf\sets{t>0}{\int_0^t
\inds{R_u\in C}\, du\geq 1}.\]
\item
An approximation to the  situation in (1) arises  when we assume that the set $C$ is of the form
$(-\eps, \eps)$ for a constant $\eps>0$. If $\eps$ is small enough the occupation time
$\int_0^t \inds{R_u\in C}\, du$ can be well approximated by the scaled local time
$\frac{1}{2\eps} l^{R}_t$ of the process $R_t$ at $0$. Thus, we may set $\kappa_t=1 \wedge l^{R}_t$.
 An instance of such a {\bf local-time driven} example will be treated
explicitly in Section \ref{sec:exam}.
\end{enumerate}
\end{exam}
\subsection{The Optimization Problem} Having introduced the notion of  the utility functional, we turn to
the statement of our central optimization problem and we call it
 the {\bf Primal Problem}. We describe it in terms of its  value function
$u:\R\to\R$ as follows
\begin{equation}
    %\nonumber
    \label{equ:problem}
    \begin{split}
        u(x)\triangleq \sup_{c\in\AA(x)} \fU(c), \quad x\in\R,
    \end{split}
\end{equation}
where $\AA(x)$ denotes the set of all
$d\kappa$-{\em densities} of $(x,\EN)$-financeable consumption
processes. Since we shall be working exclusively with consumption
processes admitting a $d\kappa$-density, no ambiguities should
arise from this slight abuse of notation. In order to have a non-trivial optimization problem, we
impose the following standard assumption:
\begin{ass}\label{ass:finite}
There exists a constant $x>0$ such that $u(x)<\infty$.
\end{ass}
\begin{rem}\  \begin{itemize}
\item[(1)] The Assumption \ref{ass:finite} is, of course,
non-trivial, although quite common in the literature. In general,
it has to be checked on a case-by-case basis. In the particular
case, when the stock-price process is an It\^ o process on a
Brownian filtration with bounded coefficients, the Assumption
\ref{ass:finite} is satisfied when there exist constants $M>0$ and
$\ld<1$ such that
\[ 0\leq U(t,x)\leq M(1+x^{\ld}),\ \text{for all $(t,x)\in
[0,T]\times (0,\infty)$.}\] For reference see \cite{KarShr98}, p.
274, Remark 3.9.
\item[(2)] Part\label{rem:constfinite} (\ref{U2}) of the Definition
\ref{def:utilfield} of a utility random field implies that
$\fU(c)\in (-\infty, \infty)$ for any constant consumption process
$\prf{c_t}$, i.e. a process $\prf{c_t}$ such that $c_t\equiv x$
for some constant $x>0$. It follows that $u(x)>-\infty$ for all
$x>0$.
\end{itemize}
\end{rem}

%functional analysis

\section{The Functional-Analytic Setup}\label{sec:funanal}
In this section we introduce several linear spaces of stochastic processes and finitely-additive
measures. They  will
prove indispensable in the convex-duality treatment of the
optimization problem defined in (\ref{equ:problem}).

\subsection{Some Families of Finitely-Additive Measures}
Let $\OO$ denote the $\sigma$-algebra of optional sets relative to the
filtration $\prf{\FF_t}$. A measure $\QQ$ defined on
$\FF_T$, and absolutely continuous to $\PP$ induces a measure $\QQk$ on $\OO$, if we set
\begin{equation}\label{equ:defqk}
 \QQk[A]=\eqitk{\ind{A}(t,\omega)},\ \text{for}\ A\in\OO.
 \end{equation}
For notational clarity, we shall always identify optional
stochastic processes $\prf{c_t}$ and random variables $c$ defined
on the product space $[0,T]\times\Omega$ measurable with respect
to the optional $\sigma$-algebra $\OO$. Thus, the measure $\QQk$
can be seen as acting on an optional processes by means of
integration over $[0,T]\times \Omega$ in the Lebesgue sense. In that spirit we
introduce the following notation
\begin{equation}
\label{equ:qnot} \scl{c}{\QQ}\triangleq \int_{[0,T]\times\Omega}
c\  d\QQ,
\end{equation}
 for a measure $\QQ$ on the optional $\sigma$-algebra $\OO$, and an optional process $c$ whenever the
 defining integral exists.
 A useful representation of the action $\scl{c}{\QQk}$ of $\QQk$ on an optional
 process $\prf{c_t}$ is given in the following proposition.
 \begin{prop}\label{pro:qmart}
\label{pro:intdual}
    Let $\QQ$ be a measure on $\FF_T$, absolutely continuous with respect to $\PP$. For a nonnegative optional process $\prf{c_t}$
    we have
\begin{equation}
    \nonumber
    \label{equ:intdual}
    \begin{split}
        \scl{c}{\QQk}=\EE\int_0^T c_t Y^{\QQ}_t\, d\kappa_t,
    \end{split}
\end{equation}
where $\prf{\yq_t}$ is the \cd version of the
martingale $\prf{\EE[\frac{d\QQ}{d\PP}|\FF_t]}$.
\end{prop}
\begin{proof}
 Define a nondecreasing \cd process $\prf{C_t}$, by
  $C_t\triangleq \int_0^t c_u\,d\kappa_u.$
  By the integration-by-parts formula we have
\[
  \yq_{\tT} C_{\tT}=
  \int_0^{\tT} \yq_{t-}\, dC_t
  +\int_0^{\tT} C_{t-}\, d\yq_t
  +\sum_{0\leq t\leq {\tT}}\Delta \yq_t \Delta C_t
  = \int_0^{\tT} \yq_{t}\, dC_t
  +\int_0^{\tT} C_{t-}\, d\yq_t,
\]
for every stopping time $\tT\leq T$. By (\cite{Pro90}, Theorem
III.17, page 107), the process $\prf{\int_0^t C_{u-}\, d\yq_u}$ is a
local martingale, so we can find an increasing  sequence of
stopping times $\sq{\tau_n}$, satisfying $\PP[\tau_n<T]\to 0$, as
$n\to \infty$, such that $\EE\int_0^{\tau_n} C_{t-}\, d\yq_t=0$,
for every $n\in\N$. Taking expectations and letting $n\to\infty$,
Monotone Convergence Theorem implies that
\begin{equation}
    \nonumber
    %\label{equ:none}
    \begin{split}
     \scl{c}{\QQ^{\kappa}}&=\EE^{\QQ}[C_T]=\EE[\yq_T C_T]
  =\lim_{n\to\infty}\EE\int_0^{\tau_n}\yq_t\, dC_t
  =\EE\int_0^T \yq_t\, dC_t\\ &=\EE\int_0^T c_t\yq_t\, d\kappa_t.
    \end{split}
\end{equation}
\end{proof}

\begin{rem}
Note that the advantage of Proposition \ref{pro:intdual} over an invocation of the Radon-Nikodym
 theorem is in the fact that the version obtained by the
 Radon-Nikodym derivative is  merely optional, and not necessarily \cd.
\end{rem}

We define $\MMk\triangleq\sets{\QQk}{\QQ\in\MM}$. The set $\MMk$
corresponds naturally to the set of all martingale measures in our
setting, and considering measures on the product space $[0,T]\times \Omega$ instead
of the measures on $\FF_T$ is indispensable for utility maximization with stochastic clock.
Most of the existing approaches to optimal consumption
start with equivalent martingale measures on $\FF_T$ and relate them
the to stochastic processes on $\prf{\FF_t}$ through some process of regularization.
In our setting, the generic structure of the stochastic clock $\prf{\kappa_t}$ renders such a line
of attack impossible.

However, as it will turn out, $\MMk$ is too small for
duality treatment of the utility maximization problem. We
shall need to enlarge it so as to contain finitely-additive
along with the countably additive measures. To make headway
with this enlargement, we consider the set of all bounded
finitely-additive measures $\QQ$ on $\OO$, such that $\PPk[A]=0$
implies $\QQ[A]=0$, and we denote this set by $\ba(\OO,\PPk)$.
It is well known that $\ba(\OO,\PPk)$, supplied
with the total-variation norm, constitutes a Banach space which is
isometrically isomorphic to the topological dual of
$\linf(\OO,\PPk)$ (see \citet{DunSch88} or \citet{RaoRao83}). The
action of an element $\QQ\in\ba(\OO,\PPk)$ on
$c\in\linf(\OO,\PPk)$ will be denoted by $\scl{c}{\QQ}$ - a
notation that naturally supplements the one introduced in (\ref{equ:qnot})

 On the Banach space $\ba(\OO,\PPk)$ there is a
canonical partial ordering transferred from the pointwise order of
$\linf(\OO,\PPk)$, equipping it with the structure of a Banach
lattice. The positive orthant of $\ba(\OO,\PPk)$ will be
denoted by $\ba(\OO,\PPk)_+$. An element $\QQ\in\ba(\OO,\PPk)_+$ is said to
be {\bf purely finitely-additive} or {\bf singular} if there exist no
nontrivial
countably additive $\QQ'\in\ba(\OO,\PPk)_+$ such that
$\QQ'[A]\leq \QQ[A]$ for all $\AA\in\OO$.
It is the content of the
Yosida-Hewitt decomposition (see \citet{YosHew52}) that each
$\QQ\in\ba(\OO,\PPk)_+$ can be uniquely decomposed as
$\QQ=\QQ^r+\QQ^s$, with $\QQ^r, \QQ^s\in\ba(\OO,\PPk)_+$,  where
$\QQ^r$ is a $\sigma$-additive measure, and $\QQ^s$ is purely
finitely-additive.

Having defined the ambient space $\ba(\OO,\PPk)$, we turn our
attention to the definition of the set $\DDk$ which will serve as a
building block in the advertised enlargement of the set $\MMk$.
Let $(\MMk)^{\circ}$ be the polar of $\MMk$ in $\linf(\OO,\PPk)$,
and let $\DDk$ be the polar of $(\MMk)^{\circ}$ (the bipolar of
$\MMk$), i.e.
\begin{equation}
    \nonumber
    \label{equ:polars}
    \begin{split}
    (\MMk)^{\circ}&\triangleq \set{c\in\linf(\OO,\PPk)\,:\,\scl{c}{\QQ}\leq 1,\ \text{for all}\ \QQ\in\MMk}.\
     \\
\DDk&\triangleq \set{\QQ\in\ba(\OO,\PPk)\,:\,\scl{c}{\QQ}\leq 1,\
\text{for all}\ c \in(\MMk)^{\circ}},
    \end{split}
\end{equation}
and we note immediately that $\DDk\subseteq\ba(\OO,\PPk)_+$,
because $ (\MMk)^{\circ}$ contains the negative orthant $-\linf_+(\OO,\PPk)$ of $\linf(\OO,\PPk)$.

Finally, for $y>0$ we define \[\MMk(y)\triangleq \set{\xi\QQ \,:\,
\xi\in [0,y],\,\QQ\in\MMk},\ \text{and} \ \DDk(y)\triangleq
\set{y\QQ \,:\,  \QQ\in\DDk}. \] Observe that $\MMk(y)\subseteq
\DDk(y)$ for each $y\geq 0$. Even though $\MMk(y)$ will typically
be a proper subset of $\DDk(y)$ for any $y>0$, the following
proposition shows that the difference is, in a
sense, small.
\begin{prop} \label{pro:mdensd} For $y>0$, $\MMk(y)$ is $\sigma(\ba(\OO,\PPk),\linf(\OO,\PPk))$-dense
in $\DDk(y)$.
\end{prop}
\begin{proof}
It is enough to provide a proof in the case $y=1$. We start  by
showing that $\DDk(1)$ is contained in the
$\sigma(\ba(\OO,\PPk),\linf(\OO,\PPk))$ - closure
$\cl{\MMk-\ba(\OO,\PPk)_+}$ of the set
$\MMk-\ba(\OO,\PPk)_+$,  where
\[\MMk-\ba(\OO,\PPk)_+\triangleq\set{\QQ-\QQ'\,:\,
\QQ\in\MMk,\,\QQ'\in\ba(\OO,\PPk)_+}.\] Suppose, to the contrary,
 that there exists $\QQ^*\in\DDk(1) \setminus
\cl{\MMk-\ba(\OO,\PPk)_+}$.
 By the Hahn-Banach theorem
there will exist an element $c^*\in\linf(\OO,\PPk)$, and constants
$a<b$ such that $\scl{c^*}{\QQ^*}\geq b$ and $\scl{c^*}{\QQ}\leq
a$, for all $\QQ\in\cl{\MMk-\ba(\OO,\PPk)_+}$. Since
$\MMk-\ba(\OO,\PPk)_+$ contains all negative elements of
$\ba(\OO,\PPk)$, we conclude that $c^*\geq 0$, $\PPk$-a.s. and
so, $0\leq a$. Furthermore, the positivity of $b$ implies that
$\PPk[c^*>0]>0$, since the probability measures in $\MMk$ are
equivalent to $\PPk$. Therefore, $0<a<b$, and the random variable
$\frac{1}{a} c^*$ belongs to $(\MMk)^{\circ}$.  It follows that
$\scl{c^*}{\QQ^*}\leq a$, a contradiction with fact that
$\scl{c^*}{\QQ^*}\geq b$.

To finalize the proof we pick $\QQ\in\DDk'(1)\triangleq
\set{\QQ\in\DDk(1)\,:\, \scl{1}{\QQ}=1}$ and take a directed set
$A$ and a net $(\tilde{\QQ}_{\alpha})_{\alpha\in A}$ in
$\MMk-\ba(\OO,\PPk)_+$ such that $\tilde{\QQ}_{\alpha}\to \QQ$.
Such a net exists thanks to the result of the first part of this
proof.
 Each $\tilde{\QQ}_{\alpha}$ can be written as
$\tilde{\QQ}_{\alpha}=\QQ^{\MMk}_{\alpha}-\QQ^{+}_{\alpha}$ with
$\QQ^{\MMk}_{\alpha}\in\MMk$ and
$\QQ^{+}_{\alpha}\in\ba(\OO,\PPk)_+$, for all $\alpha\in A$.
Weak-* convergence of the net $\tilde{\QQ}_{\alpha}$ implies that
$\scl{1}{\QQ^{+}_\alpha}\to 0$ and therefore $\QQ^{+}_\alpha\to 0$
in the norm- and weak-* topologies. Thus $\QQ^{\MMk}\to \QQ$ and
we conclude that $\MMk$ is dense in $\DDk'(1)$. It follows
immediately that $\MMk(1)$ is dense in $\DDk(1)$.
\end{proof}

\subsection{The space $\slzmk$}
Let $\slzmk$ stand for the vector space of all optional random
processes $\prf{c_t}$ verifying
\[\normm{c}<\infty,\ \text{where}\
\normm{c}\triangleq \sup_{\QQ\in\MMk}\scl{\abs{c}}{\QQ}.\] It is
quite clear that $\normm{\cdot}$ defines a norm on $\slzmk$. We
establish completeness in the following proposition.
\begin{prop}
    \label{pro:isbanach}
    $(\slzmk,\normm{\cdot})$ is a Banach space.
\end{prop}
\begin{proof}
    To prove
    that $\slzmk$ is complete under $\normm{\cdot}$, we take a sequence $\sq{c_n}$ in $\slzmk$ such that
    $\sum_{n} \normm{c_n }<\infty$. Given a fixed, but arbitrary
    ${\tilde{\QQ}_{\kappa}}\in\MMk$, the inequality $\normm{c}\geq
    \scl{|c|}{{\tilde{\QQ}_{\kappa}}}$ holds for every $c\in\slzmk$ and thus
    the series $\sum_{n=1}^{\infty}\abs{c_n}$ converges in $\lone(\OO,{\tilde{\QQ}_{\kappa}})$.
    We can, therefore,
     find an optional process
    $c_0\in\lone({\tilde{\QQ}_{\kappa}},\OO)$ such that $c_0=\lim_{n\to\infty} \sum_{k=1}^n c_k$,
    in $\lone({\tilde{\QQ}_{\kappa}},\OO)$
    and ${\tilde{\QQ}_{\kappa}}$-almost surely.

   For an arbitrary $\QQk\in\MMk$ we
    have: \[
    \scl{
        |
            c-\sum_{k=1}^n c_k
        |}{\QQk}
     \leq
    \sum_{k=n+1}^{\infty} \scl{\abs{c_k}}{\QQk} \leq
    \sum_{k=n+1}^{\infty} \normm{c_k}.\]
    By taking the supremum over all $\QQk\in\MMk$, it follows that $c_0\in\slzm$ and $\sum_{k=1}^{\infty} c_k=c_0$
    in $\normm{\cdot}$.
\end{proof}

\begin{rem}
A norm of the form $\norm{\cdot}_{\MM}$ has first appeared in
\citet{DelSch97a}, where the authors study the Banach-space
properties of the space of {\em workable contingent claims}.
\end{rem}
 At this point, we can introduce the third (and final) update of the
 notation of (\ref{equ:qnot}). Let  $\slzmkp$
 denotes the set of nonnegative elements in $\slzmk$. For $c\in\slzmkp$
 a constant $y>0$ and $\QQ\in\DDk(y)$, we define
\begin{equation}
    %\nonumber
    \label{equ:action}
    \begin{split}
     \scl{c}{\QQ}\triangleq
\sup\sets{\scl{c'}{\QQ}}{ c'\in\linf(\OO,\PPk)_+,\ c'\leq c \,
\text{ $\PPk$-a.s.}}.
    \end{split}
\end{equation}
  Proposition \ref{pro:mdensd} implies that
$\scl{c}{\QQ}\leq y \normm{c}<\infty$ for any $\QQ\in\DDk(y)$. We
can therefore extend the mapping $\scl{\cdot}{\cdot}$ to a pairing
(a bilinear form) between the vector spaces $\slzmk$ and $\bam$, where
$\bam$ is defined as the linear space spanned by $\DDk$, i.e.
\[ \bam\triangleq \set{\QQ\in\ba(\OO,\PPk)\,:\,
\exists\,y>0, \QQ^+,\QQ^- \in\DDk(y)\ \text{such that} \
\QQ=\QQ^+-\QQ^-}.\] The linear space $\bam$ plays the role of the
ambient space in which the dual domain will be situated. It will
replace the space $\ba$ appearing in \cite{CviSchWan01} and
\cite{KarZit03}, and allow us to deal with unbounded random
endowment and the stochastic clock.

In this way the action $\scl{\cdot}{\QQ}$ defined in (\ref{equ:action})
identifies $\QQ\in\bam$
with a linear functional on $(\slzm,\normm{\cdot})$, and by the
construction of the pairing $\scl{\cdot}{\cdot}$, the dual norm
\[ \norm{\QQ}_{\bam}\triangleq
\sup_{c\in\slzmk\,:\, \normm{c}\leq 1}\abs{\scl{c}{\QQ}}\]
 of $\QQ\in\DDk(y)$ (seen as
a linear functional on $\slzmk$) is at most equal to $2y$. We can,
therefore, identify $\bam$ with a subspace of the topological dual
of $\slzmk$ and $\DDk(y)$ with its bounded subset. Moreover, by the virtue of its definition as a polar
set of $(\MMk)^{\circ}$, $\DDk(y)$ is closed in $\bam$ in the
$\sigma(\bam,\slzmk)$-topology, so that the following proposition becomes
is a direct consequence of Alaoglu's Theorem
\begin{prop} \label{pro:iscomp} For every $y>0$, $\DDk(y)$ is
$\sigma(\bam,\slzmk)$-compact.
\end{prop}

Finally, we state a version of the budget-constraint
characterization of admissible consumption processes, rewritten
to achieve a closer match with our newly introduced setup. It
follows directly from Propositions \ref{pro:charadm} and
\ref{pro:mdensd}.

\begin{prop}
    \label{pro:charadm2}
    For any $y>0$, $x\in\R$ and a nonnegative optional process $\prf{c_t}$, we have the following equivalence
    \[ c\in\AA(x,\EN)\ \iff  y\scl{c}{\QQ}\leq xy+\scl{e}{\QQ}\ \text{for all $\QQ\in \DDk(y)$,}\]
    where $\EN_t=\int_0^t e_u\,d\kappa_u$.
    Moreover, to check whether $c\in \AA(x,\EN)$, it is enough to
    show $y\scl{c}{\QQ}\leq xy+\scl{e}{\QQ}$ for all $\QQ\in\MMk(y)$
    only.
\end{prop}

% Utility

\section{The Dual Optimization Problem and the Main Result}\label{sec:fv}

\subsection{The Convex Conjugate $\fV$ and Related Functionals} We define  a convex functional
  $\fV:\bam \to (-\infty,\infty]$, by
\begin{equation}
\label{equ:defV}
 \fV(\QQ)\triangleq \sup_{c\in\slzm_+} \Big( \fU(c)-\scl{c}{\QQ} \Big),
 \end{equation}
 and call it the {\bf convex conjugate of $\fV$}. The functional $\fV$ will play the central role in
 the convex-duality treatment of our utility-maximization problem.

By strict concavity and continuous differentiability of the
mapping $x\mapsto U(\omega, t, x)$, there exists a unique random
field $I:\Omega\times [0,T]\times (0,\infty)$ that solves the
equation $U_x(\omega,t, I(\omega, t, y))=y.$
Using the random field $I$, we introduce a functional $\fI$,
defined on and taking values in  the set of strictly positive
optional process, by $\fI(Y)_t(\omega)=I(\omega, t, Y_t)$. The
functional $\fI$ is called {\bf the inverse marginal utility
functional}. We note for the future use the well-known
relationship
\begin{equation}
    %\nonumber
    \label{equ:conj}
    \begin{split}
U(\omega, t, I(\omega, t, y))=V(\omega, t, y)+yI(\omega, t, y),\
(\omega,t,y)\in \Omega\times [0,T]\times (0,\infty),
    \end{split}
\end{equation}
where $V$ is {\bf the convex conjugate} of the utility random
field $U$, defined by $V(\omega, t,y)\triangleq\sup_{x>0}
[U(\omega, t, x)- xy]$, for $(\omega,t,y)\in \Omega\times
[0,T]\times (0,\infty)$.

 For a function $f:X\to\bar{\R}$ with an arbitrary domain $X$,
 taking values in the extended set of real numbers $\bar{\R}=[-\infty,\infty]$, we adopt the standard
 notation $\Dom(f)=\sets{x\in X}{f(x)\in (-\infty,\infty)}$.

The following proposition represents the convex conjugate $\fV$ in
terms of the regular part of its argument, relating the definition
(\ref{equ:defV}) to the corresponding formulations in
\cite{CviSchWan01} and \cite{KarZit03}.
\begin{prop} \label{pro:propfv} The domain $\Dom(\fV)$ of the convex
conjugate $\fV$ of $\fU$ satisfies $\Dom(\fV)\subseteq\bam_+$, and
$\Dom(\fV)+\bam_+\subseteq\Dom(\fV)$. For $\QQ\in\Dom(\fV)$, we
have $\fV(\QQ)=\fV(\QQ^r)$, where $\QQ^r\in\bam_+$ is the regular
part of the finitely-additive measure $\QQ$. Moreover, there
exists a non-negative optional process $\yq$, such that
\begin{equation}
    %\nonumber
    \label{equ:yrep}
    \begin{split}
     \fV(\QQ)=\EE\int_0^T V(t,\yq_t)\, d\kappa_t.
    \end{split}
\end{equation}
When $\QQ$ is countably-additive, the process $\prf{\yq_t}$
coincides with the synonymous martingale defined in Proposition
\ref{pro:qmart}.
\end{prop}
\begin{proof} For $\QQ\not\in\bam_+$, there exists an optional set $A$
such that $q\triangleq -\QQ[A]>0$. For a constant $\eps>0$, we
define a sequence $\sq{c^n}$ of optional processes by
$c^n\triangleq \eps+n\ind{A}$. Let $G$ being the constant from
Definition \ref{def:utilfield} (1)(b). Then
\begin{equation}
    \nonumber
    %\label{equ:none}
    \begin{split}
         \fV(\QQ)\geq\fU(c^n)-\scl{c^n}{\QQ} \geq \eitk{U(\omega, t,\eps)}-\eps+nq
 \geq G-\eps+nq \to \infty,
    \end{split}
\end{equation}
yields $\fV(\QQ)=\infty$,
 and so $\Dom(\fV)\subseteq\bam_+$. To show that $\Dom(\fV)+\bam_+\subseteq \Dom(\fV)$ we only need
 to note that it follows directly from the monotonicity of $\fV$.

 For the second claim, let $\QQ\in\bam_+$ and let $\Sing(\QQ)$ denote the family of
 all optional sets $A\subseteq [0,T]\times\Omega$  such that $\QQ^s(A)=0$, where $\QQ^s$
 denotes the singular part of the finitely-additive measure $\QQ$.
     For $A\in \Sing(\QQ)$, $\delta > 0$, and an arbitrary $c\in \slzm_+$,
     we define an optional process $\hat{c}=\hat{c}^{(\delta,A)}$
    by $\hat{c}\triangleq c\ind{A}+\delta c \ind{A^c}$.
    Excluding the trivial cases when $\fU(c)=-\infty$ or $\fU(c)=+\infty$, we assume
    $\fU(c)\in \R$, so that Assumption \ref{ass:delta} implies that
    $\fU(\delta c), \fU(\hat{c}) \in\R$, as well. Now
\begin{equation}
    \label{equ:regVproof}
    \begin{split}
        &\fU(c)-\scl{c}{\QQ^r}-\fU(\hat{c})+\scl{\hat{c}}{\QQ}=\\ &
        \eitk{\big( U(t,c_t)-U(t,\delta c_t) \big)\ind{A^c}}-(1-\delta)\scl{c\ind{A^c}}{\QQ^r}+\delta\scl{c}{\QQ^s}.
    \end{split}
\end{equation}
According to  \citet{RaoRao83} (Theorem 10.3.2, p. 234),
$\Sing(\QQk)$
    contains sets with the $\PPk$-probability arbitrarily close to $1$, so we can make the right-hand side of the
    expression in (\ref{equ:regVproof}) arbitrarily small in absolute value, by a suitable choice of $A\in\Sing(\QQ)$ and $\delta$.
It follows immediately that
\[ \fV(\QQ^r)=\sup_{c\in\slzm} [ \fU(c)-\scl{c}{\QQ^r}]\leq  \sup_{c\in\slzm} [ \fU(c)-\scl{c}{\QQ}]
=\fV(\QQ),\] and the equality $\fV(\QQ)=\fV(\QQ^r)$ follows from
the monotonicity of $\fV$.

Note further that $\QQ^r$ is a countably-additive measure on the
$\sigma$-algebra of optional sets, absolutely continuous with
respect to the measure $\PPk$. It follows by the Radon-Nikodym
theorem that optional process $\prf{\yq_t}$ defined by
\begin{equation}\label{equ:yrep2}
 \yq(t,\omega)=\frac{d\QQ^r}{d\PPk},\ \text{satisfies}\
 \scl{c}{\QQ^r}=\EE\int_0^T c_t\yq_t \, d\kappa_t.
\end{equation}

Let us combine now the representation (\ref{equ:yrep2}) with the
fact that $\fV(\QQ)=\fV(\QQ^r)$. By the definition of the convex
conjugate function $V$,
\begin{equation}
    \nonumber
    %\label{equ:none}
    \begin{split}
\fV(\QQ) & =\fV(\QQ^r)=\sup_{c\in\slzm_+} ( \fU(c)-\scl{c}{\QQ^r})
\\ &=\sup_{c\in\slzm_+} \EE\int_0^T \Big( U(t,c(t))-c(t)\yq_t \Big)
\, d\kappa_t
 \leq \EE\int_0^T V(t,\yq_t) \,d\kappa_t
    \end{split}
\end{equation}
The reverse inequality follows from the differentiability of the
function $V(t,\cdot)$ by taking a bounded sequence in $\slzm$
converging to $-\frac{\partial}{\partial y} V(t,y)$ monotonically,
in the supremum defining $\fV(\QQ^r)$.
\end{proof}

\begin{rem}
The action of the functional $\fI$ can be extended to the set of
all $\QQ\in\bam_+$ satisfying $\yq_t
> 0$ $\PPk$-a.e. by
$ \fI(\QQ)_t\triangleq  \fI(\yq)_t$, obtaining immediately
$\fI(\QQ)=\fI(\QQ^r)$.
\end{rem}

\subsection{The Dual Problem} The convex conjugate $\fV$ will serve as the main
ingredient in the convex-duality treatment of the Primal Problem.
We start by introducing  the {\bf Dual Problem}, with the value
function $v$:
\begin{equation}
    %\nonumber
    \label{equ:adp}
    \begin{split}
        v(y)\triangleq \inf_{\QQ\in \DDk(y)} \fVE(\QQ),\quad y\in [0,\infty) ,\ \text{where}\
        \fVE(\QQ)\triangleq \fV(\QQ)+\scl{e}{\QQ}.
    \end{split}
\end{equation}
For $y<0$ we set $v(y)=+\infty$, and note that $v(0)<\infty$
precisely when the utility functional $\fU$ is bounded from above.

\subsection{The Main Result}
Finally we state our central result in the following theorem. The
proof will be given through a number of auxiliary results in
Appendix A.
\begin{thm} \label{thm:aer} Let the financial market $\prf{S^i_t}$, $i=1,\dots, d$ be arbitrage-free as in Assumption
\ref{ass:NFLVR}, and let the random endowment process
$\prf{\EN_t}$ admit a density $\prf{e_t}$ so that $\EN_t=\int_0^t
e_u\, d\kappa_u$, where $\prf{\kappa_t}\in\slzer_+$ is a
stochastic clock. Let $U$ be a utility random field as defined in
\ref{def:utilfield} and $\fU$ the corresponding utility
functional. If $\fU$ satisfies Assumption \ref{ass:delta} and the
value function $u$ satisfies Assumption \ref{ass:finite}, then
\begin{enumerate}
  \item the concave value function
   $u(\cdot)$ is finite and strictly increasing on
   $(-\LL(\EN),\infty)$, and $u(x)=-\infty$ for $x<-\LL(\EN)$, where
   $\LL(\EN)\triangleq \inf_{\QQ\in\MM} \EE^{\QQ}[\EN_T]$ denotes the lower
   hedging price of the contingent claim $\EN_T$.
   \item
   $\lim_{x\to (-\LL(\EN))+} u'(x)=+\infty$ and $\lim_{x\to \infty}
   u'(x)=0$.
  \item The dual value function $v(\cdot)$ is finitely valued and continuously differentiable
  on $(0,\infty)$ and $v(y)=+\infty$ for $y< 0$.
  \item $\lim_{y\to 0+} v'(y)=-\infty$ and $\lim_{y\to\infty}
  v'(y)=-\LL(\EN)$.
 \item For any $y\geq 0$, there exists a solution to the
   Dual problem (\ref{equ:adp}) - i.e.
  $v(y)=\fV(\hqy)+\scl{e}{\hqy}$, for some $\hqy\in\DDk(y)$.
  \item For $x>-\LL(\EN)$ the Primal Problem has a solution
  $\prf{\hcx_t}$, unique $d\kappa$-a.e.
  \item The unique solution $\prf{\hcx_t}$ of the primal problem is of the form $\hcx_t=\fI(\hqy)_t$
  where $\hqy$ is a solution of the dual problem corresponding to
  $y>0$ such that $x=-v'(y)$.
\end{enumerate}
\end{thm}

\subsection{A Closer Look at the Dual Domain} Given that the solution of the
Primal problem can be expressed as a function of the process
$\prf{\yq_t}$ from Proposition \ref{pro:propfv}, it will be useful
to have more information on its probabilistic structure. When
$\QQ\in\MMk$, Proposition \ref{pro:intdual} implies that $\yq$ is
a nonnegative \cd martingale. In general, we can only establish
the supermartingale property for a (large enough) subclass of
($\PPk$-a.s.)-maximal processes in $\sets{\yq}{ \QQ\in\DD(1)}$. In
the contrast with the case studied in \citet{KarZit03}, we cannot
establish any strong trajectory regularity properties such as
right-continuity, and will only have to satisfy ourselves with the
weaker property of optional measurability.

\begin{prop}
\label{pro:super} For $\QQ\in\DD(1)$ there exist an optional
process $\prf{F_t}$, taking values in $[0,1]$, and $\QQ'\in\DD(1)$
such that
\begin{enumerate}

\item $\yq_t=Y^{\QQ'}_t F_t$,

\item The process $\prf{Y^{\QQ'}_t}$ has a $d\kappa$-version which is an
optional supermartingale, and

\item there exists a sequence of
martingale measures $\seq{\QQ}$ such that $Y^{\QQ_n}\to Y^{\QQ'}$,
$d\kappa$-a.e.
\end{enumerate}
\end{prop}

\begin{proof} We start by
observing that $\EE[\int_0^T \yq_t c(t)\, d\kappa_t]\leq
\scl{c}{\QQ}\leq 1$, for all $c\in\AA(1,0)$. In other words, $\yq$
is in the $\PPk$-polar set of $\AA(1,0)$, in the terminology of
\citet{BraSch99}. By characterization in Proposition
\ref{pro:charadm2}, $\AA(1,0)$ can be written as the polar of
$\MMk$, and the Bipolar Theorem of \citet{BraSch99} states
that $\yq$ is an element of the smallest convex, solid and closed
(in $\PPk$-probability) set containing $\MMk$. Therefore, there
exists a process $\prf{F_t}$, taking values in $[0,1]$, and an
optional process $\prf{Y_t}$, ($\PPk$-a.s.)-maximal in the bipolar
of $\MMk$, such that $\yq_t=Y_t F_t$. Moreover, the same theorem
implies that there exists a
sequence $\sequ{\QQ}$ in $\MM$, and a sequence $\sequ{F}$ of
optional processes taking values in $[0,1]$, such that $\yqn_t
F^{(n)}_t\to Y_t$, $\PPk$ a.s. The sequence of positive processes
$\yqn$ is bounded in $\lone(\PPk)$ and thus the theorem of
Koml{\'o}s (see \citet{Sch86}) asserts existence of a nonnegative
optional process $\prf{\tilde{Y}_t}$, and a sequence of finite
convex combinations of the elements of the sequence $\sequ{\QQ}$
(still denoted by $\sequ{\QQ}$) such that $\yqn_t\to \tilde{Y}_t$
$\PPk$-a.s. It is now a simple consequence of Fatou's lemma that
$\tilde{Y}$ is an element of the bipolar of $\MMk$ dominating
$Y_t$. Since $Y_t$ is maximal, we conclude that
$\tilde{Y}_t=Y_t$ $\PPk$-a.s.  The supermartingale property of
$\prf{Y}$ follows from Fatou's lemma applied to the sequence
$\{\prf{\yqn_t}\}_{n\in\N}$.

We are left now with the task of producing $\QQ'\in\DD(1)$, such
that $Y_t=Y^{\QQ'}_t$. In order to do that, take $\QQ'$ to be any
cluster point of the sequence $\sequ{\QQ}$ in $\DD(1)$ in the
$\sigma(\bam,\slzmk)$-topology. Existence of such a $\QQ'$ is
guaranteed by Proposition \ref{pro:iscomp}. Finally, it is a
consequence of (\citet{CviSchWan01}. Lemma A.1, p. 16) that
$Y_t=Y^{\QQ'}_t$-$\PPk$-a.s.
\end{proof}

\section{An Example} \label{sec:exam}
In order to illustrate the theory developed so far, in this section we present
an example of a utility-maximization problem with a random clock given by the local time at 0 of an
Ornstein-Uhlenbeck process.
\subsection{Description of the Market Model}
Let $(B_t,W_t)_{t\in [0,\infty)}$ be two correlated Brownian
motions defined on a probability space $(\Omega, \FF,\PP)$, and
let $\prfi{\FF_t}$ be the filtration they generate, augmented by
the $\PP$-null sets in order to satisfy the usual conditions. We
assume that the correlation coefficient $\rho\in (-1,1)$ is fixed
so that $d[B,W]_t=\rho\, dt$.

The financial market will consist of one riskless asset
$S^0_t\equiv 1$, and a risky asset $\prfi{S_t}$ which satisfies
\begin{equation}
    \nonumber
    %\label{equ:none}
    \begin{split}
     dS_t=S_t\Big( \mu\,dt + \sigma\, dB_t\Big),\ S_0=s_0,
    \end{split}
\end{equation}
where $\mu\in\R$ is the {\bf stock appreciation rate} and
$\sigma>0$ is the {\bf volatility}.

Apart from the tradeable asset $\prfi{S_t}$, there is an
Orstein-Uhlenbeck process $\prfi{R_t}$ defined as the unique
strong solution of
\begin{equation}
    \nonumber
    %\label{equ:none}
    \begin{split}
     dR_t=-\alpha R_t\, dt+\,dW_t,\ R_0=0.
    \end{split}
\end{equation}
We call $\prfi{R_t}$ the {\bf index process}, and interpret it  as
the process modelling a certain state-variable of the economy,
possibly related to the political stability, or some aspect of the
goverment's economic policy. The index process is non-tradable and
its role is to impose constraints on the consumption: we are
allowed to withdraw money from the trading account only when
$\abs{R_t}<\eps$. An agent with an initial endowment $x$ and a
utility random field $U(\cdot,\cdot,\cdot)$ will then naturally
try to choose a strategy so as to maximize the utility of
consumption of the form
\begin{equation}
    \label{equ:exput}
    \begin{split}
     \EE \int_0^{\tau} U(\omega, t, c(t)) \inds{\abs{R_t}<\eps}\,
     dt,
    \end{split}
\end{equation}
on some trading horizon $[0,\tau]$. If we introduce the notation
$\kappa^{\eps}_t=\frac{1}{\eps} \int_0^t \inds{\abs{R_t}<\eps}\,
dt $, the expression in (\ref{equ:exput}) becomes (up to a
multiplicative constant)
\begin{equation}
    \label{equ:exput2}
    \begin{split}
      \EE \int_0^{\tau} U(\omega, t, c(t)) d\kappa^{\eps}_t.
    \end{split}
\end{equation}
Assuming that $\eps$ is a small constant, the process
$\kappa^{\eps}$ can be approximated by the local time $\kappa_t$
of the process $R_t$. We define the time horizon $\tau=\tau_1$,
where $\tau_s\triangleq \inf\sets{t>0}{\kappa_t>s}$ is the {\bf
inverse local time process}. In this way our agent will get
exactly one unit of {\bf consumption time} (as measured by the
clock $\kappa$) from the start to the end of the trading interval.
It will, therefore, be our goal to solve the following problem,
defined in terms of its value function $u(\cdot)$:
\begin{equation}
    %\nonumber
    \label{equ:probi}
    \begin{split}
     u(x)=\sup_{c\in\AA(x,0)} \EE\int_0^{\tau_1}U(\omega, t,c_t)\, d\kappa_t,\ x>0.
    \end{split}
\end{equation}

\subsection{Absence of Arbitrage} The time-horizon $\tau$ defined
above is clearly not a bounded random variable, so the
results in the main body of this paper do not  apply directly.
However, in order to pass from an infinite to a finite horizon, it
is enough to apply a deterministic time-change that maps
$[0,\infty)$ onto $[0,1)$ and note that no important part of the
structure of the problem is lost in this way (we leave the easy
details of the argument to the reader). Of course, we need to show
that all the assumptions of Theorem \ref{thm:aer} are satisfied.
The validity of Assumption \ref{ass:finite} will have to be
checked on a case-by-case basis (see Remark \ref{rem:asslog},
for the case of $\log$-utility). Therefore, we are
left with Assumption \ref{ass:NFLVR}. In order to proceed we need
to exhibit a countably-additive probability measure $\QQ$
equivalent to $\PP$ such that the asset-price process $\prfi{S_t}$
is a $\QQ$-local martingale on the stochastic interval
$[0,\tau_1]$. The obvious candidate will be the measure $\QQ_0$
defined in terms of its Radon-Nikodym derivative with respect to
$\PP$, by
\begin{equation}
    %\nonumber
    \label{equ:cands}
    \begin{split}
     \frac{d\QQ_0}{d\PP}=Z^0_{\tau_1},\ \text{where}\
     Z^0_{\tau_1}\triangleq \exp(-\theta
     B_{\tau_1}-\frac{1}{2}\theta^2 \tau_1),
    \end{split}
\end{equation}
and $\theta=\mu/\sigma$ is the {\bf market price of risk}
coefficient. Once we show that $\EE[Z^0_{\tau_1}]=1$, it will
follow directly from Girsanov's theorem (see \citet{KarShr91},
Theorem 3.5.1, page 191.) that $\prfi{S}$ is a $\QQ$-local
martingale on $[0,\tau_1]$. The equivalence of the measures
$\QQ_0$ and $\PP$ is a consequence of the fact that $\tau_1<\infty$
a.s, which follows from the following proposition which lists some distributional properties
of the process $\prfi{R_t}$ and its local time $\prfi{\kappa_t}$.
\begin{prop}\ \label{pro:local} For $\xi<0$ and $x\geq 0$, let $H_{\xi}(x)$ denote the value of
the Hermite function
\begin{equation}
    %\nonumber
    \label{equ:inthermite}
    \begin{split}
         H_{\xi}( x)=\frac{1}{2\Gamma(-\xi)} \int_0^{\infty} e^{-s-2x\sqrt{s}} s^{-\frac{1}{2}\xi-1}\,
         ds.
    \end{split}
\end{equation}
For the Ornstein-Uhlenbeck process $\prfi{R_t}$ and the inverse $(\tau_s)_{s\in [0,\infty)}$ of
its local time at $0$ $\prfi{\kappa_t}$, we have the following explicit expressions:
\begin{enumerate}
\item
\begin{equation}
    %\nonumber
    \label{equ:laplace1}
    \begin{split}
 \EE[\exp(-\lambda \tau_s)|R_0=0]=\begin{cases} \exp(-s \psi(\lambda)), & \lambda > - \alpha \\
\infty, & \lambda \leq -\alpha \end{cases},
    \end{split}
\end{equation}
where the {\em Laplace exponent} $\psi(\lambda)$ is given by
\begin{equation}
    %\nonumber
    \label{equ:laplace2}
    \begin{split}
     \psi(\ld)=\alpha \frac{2^
     {1+
       \frac{\lambda }{\alpha }}
     \Gamma(
       \frac{1}{2}+
          \frac{\lambda }{2\alpha }
          )^2}{{\sqrt{2\pi
        }}\Gamma(
     \frac{\lambda }{\alpha })}.
    \end{split}
\end{equation}
\item With $T_0=\inf\sets{t>0}{R_t=0}$ we have,
\begin{equation}
    %\nonumber
    \label{equ:hit}
    \begin{split}
     \EE[\exp(-\ld T_0)|R_0=r]=j(\ld,\abs{r}),
    \end{split}
\end{equation}
where
\begin{equation}
    \nonumber
    %\label{equ:none}
    \begin{split}
     j(\ld,r)\triangleq 2^{\frac{\ld}{\alpha}}
     \frac{\Gamma(\frac{1+\frac{\ld}{\alpha}}{2})}
     {\Gamma(\frac{1}{2})} H_{-\frac{\ld}{\alpha}}
     \Big(\frac{r}{\sqrt{2}}\Big).
    \end{split}
\end{equation}

\end{enumerate}
\end{prop}
\begin{proof} See \citet{BorSal02}, equation (2.0.1), page 542, for (1), and
\citet{BorSal02}, equation (4.0.1), page 557 for (2). Use the identity
$D_{\zeta}(x)=2^{-\zeta/2}e^{-x^2/4}H_{\zeta}(x/\sqrt{2})$.
\end{proof}

To prove the equality $\EE[Z^0_{\tau_1}]=1$, it will be enough to
show that $\EE[\exp(\frac{1}{2} \theta^2 \tau_1)]<\infty$ by the Novikov's
criterion (\citet{KarShr91}, Proposition 3.5.12., page 198.)
Part (1) of Proposition \ref{pro:local} implies that for $\alpha> \theta^2/2$, we have
$\EE[\exp(\frac{1}{2} \theta^2 \tau_1)]<\infty$, which proves the following proposition:
\begin{prop} When $\alpha>\theta^2/2$, there is no arbitrage on the stochastic interval $[0,\tau_1]$.
\end{prop}

\subsection{The Optimal Consumption and Portfolio Choice} It has been shown in \citet{KarZit03}
that the maximal dual processes in
the context of the financial markets driven by It\^ o processes with bounded coefficients are in fact
local martingales, and their structure has been described. This result can be extended to our case as
follows.
\begin{thm} \label{thm:ito} Let the utility random field $U$ satisfy Assumptions \ref{ass:delta} and \ref{ass:finite}.
 Then, for $x>0$, there exists a predictable process $\prfi{\nux_t}$, such that the  $\PPk$-a.e.
unique solution $\prfi{\hcx_t}$ of the problem posed in (\ref{equ:probi}) is given by
$\hcx_t(\omega)=I(\omega, t,Z^{\nux}_t(\omega))$.
The process $\prfi{Z^{\nux}_t}$ is a local martingale satisfying \begin{equation}
    %\nonumber
    \label{equ:locnu}
    \begin{split}
     dZ^{\nux}_t=Z^{\nux}_t \Big( \nux_t\, dW_t - (\theta+\rho \nux_t)\, dB_t \Big),\qquad Z^{\nux}_0=y,
    \end{split}
\end{equation}
where $y>0$ is the unique solution of $-v'(y)=x$. The portfolio process $\prfi{\Hx_t}$ financing $\prfi{\hcx}$
and the process $\prfi{\nux_t}$
are given by
\begin{equation}
    %\nonumber
    \label{equ:pi}
    \begin{split}
     \Hx_t=\frac{X_t}{\sigma S_t} (\theta+\rho \nux_t)+\frac{\psi^{B}_t}{\sigma S_t Z^{\nux}_t},
     \quad \nux_t=\frac{1}{X_t Z^{\nux}_t}\psi^{W}_t,
    \end{split}
\end{equation}
where $\prfi{X_t}$ is the wealth process corresponding to $\prfi{\Hx_t}$ and $\prfi{\hcx_t}$, given by
\begin{equation}
    %\nonumber
    \label{equ:we}
    \begin{split}
dX_t=\Hx_t\, dS_t-\hcx_t\, d\kappa_t,\quad X_0=x,
    \end{split}
\end{equation}
and $\prfi{\psi^B}$ and $\prfi{\psi^W}$ are predictable processes such that
\begin{equation}
    %\nonumber
    \label{equ:psis}
    \begin{split}
xy+\int_0^{\tau_1} \psi^B_t\, dB_t+\int_0^{\tau_1} \psi^W_t\, dW_t=\int_0^{\tau_1} Z^{\nux}_t \hcx_t\, d\kappa_t.
    \end{split}
\end{equation}

\end{thm}
\begin{proof}
By Theorem \ref{thm:aer}, there exists a $\PPk$-a.e. unique
optimal consumption density $\hcx\in \AA(x,0)$ given by $
\hcx_t=I(t,\yq_t)$, for some $\QQ\in\DDk(y)$. Since $\prfi{\yq_t}$
solves the dual optimization problem, and is therefore $\PPk$-a.e.
maximal, Proposition \ref{pro:super} states that there exists a
sequence $\qnn$ in $\MM$  such that $\yqn\to\yq$ $\PPk$ a.s. By
taking a further sequence of convex combinations which exists
thanks to Koml\' os's Theorem (see \cite{Kom67}, \cite{Sch86}), we
can assume that $\yqn_T\to\yq_T$, $\PP$-a.s. and
$\yqn_t\to\yqn_t$, $\PP\times\ld$-a.e. Without going into tedious
but straightforward details, we note that it is the consequence of
continuity of local martingales on Brownian filtrations, the
Filtered Bipolar Theorem (\citet{Zit02}, Theorem 2), and Lemma
2.5, Theorem 2.10 and Proposition 4.1 in \citet{KarZit03}, that
$\prfi{\yq_t}$ possesses a $\PPk$-version of the form $\yq_t=y
Z^{\nu}_t$, where $Z^{\nu}$ is a  local martingale of the form
(\ref{equ:locnu}).

Knowing that $\hcx\in\AA(x,0)$, there exists a portfolio process
$\prfi{\Hx_t}$ such that the wealth process $\prfi{X_t}$ given by
(\ref{equ:we}) satisfies $X_{\tau_1}\geq 0$. The saturation of the
budget constraint (see Lemma \ref{lem:difv}, (2)) forces
$X_{\tau_1}=0$. It\^ o's Lemma shows that the process
\begin{equation}
    %\nonumber
    \label{equ:em}
    \begin{split}
     M_t=X_tZ^{\nu}_t+\int_0^t Z^{\nu}_u \hcx_u\, d\kappa_u
    \end{split}
\end{equation}
is a non-negative local martingale with $M_{\tau_1}=\int_0^{\tau_1} Z^{\nu}_u \hcx_u\, d\kappa_u$.
By Lemma \ref{lem:difv} (2), we have $\EE[M_{\tau_1}]=x=M_0$.  Therefore, $M$ is a martingale on $[0,\tau_1]$.
The second equality in (\ref{equ:pi}) follows by applying It\^ o's formula to (\ref{equ:em}), and
equating coefficients with the ones in the expansion (\ref{equ:psis}).
\end{proof}

\subsection{The Case of Logarithmic Utility}
In order to get explicit results, we consider now the agent whose utility function has the form $U(\omega, t, x)=\exp(-\beta t)\log(x)$, where the {\bf impatience rate} $\beta$ is a positive constant.
The expressions (\ref{equ:pi}) will prove indispensable because it is possible to get an explicit expression for the
processes $\prfi{\psi^W_t}$ and $\prfi{\psi^B_t}$ from (\ref{equ:psis}). The key feature of the logarithmic utility that will allow us
to do this is the fact that the inverse marginal utility function $I$ is given by $I(t,y)=\exp(-\beta t)/y$, so that
the right-hand side of (\ref{equ:psis}) becomes
\begin{equation}
    %\nonumber
    \label{equ:ppet}
    \begin{split}
     M_{\tau_1}\triangleq \int_0^{\tau_1} Z^{\nu}_t\hcx_t\, d\kappa_t=\int_0^{\tau_1} e^{-\beta t} d\kappa_t.
    \end{split}
\end{equation}

In order to progress with the explicit representation of the processes
$\prfi{\psi^W_t}$ and $\prfi{\psi^B_t}$ from (\ref{equ:psis}),  in the following lemma we prove
a useful fact about the {\bf
conditional $\beta$-potential} of the local time $\prfi{\kappa_t}$, i.e. the
random process $\prfi{G_t}$
defined by $ G_t\triangleq\EE[\int_0^{\tau_1} \exp(-\beta u)\, d\kappa_u|\FF_t].$
\begin{lem} A version of the process $G$ is given by
\begin{equation}
    %\nonumber
    \label{equ:prog}
    \begin{split}
 G_t=\begin{cases}
    \exp(-\beta t) j(\beta, \abs{R_t}) \frac{1-\exp(-(1-\kappa_t) \Psi(\beta))}{\Psi(\beta)}+
    \int_0^t e^{-\beta u}\, d\kappa_u, & \kappa_t\leq1 \\
    \int_0^{\tau_1} e^{-\beta u}\, d\kappa_u,& \text{$\kappa_t>1$},\end{cases}
    \end{split}
\end{equation}

    where the functions $\psi$ and $j$ are defined in (\ref{equ:laplace2}) and (\ref{equ:hit}).
\end{lem}
\begin{proof}
We start by defining a family of  stopping times
$T_0(t)=\inf\sets{u\geq t}{ R_u=0}$, and note that because
$d\kappa_u$ does not charge the complement of the zero-set of $R_t$, we have
\begin{equation}
    %\nonumber
    \label{equ:gone}
    \begin{split}
 G_t=
\EE[\int_{T_0(t)}^{\tau_1} e^{-\beta u}\, d\kappa_u\Big|\sigma(\kappa_t, R_t)]+\int_0^t e^{-\beta u}\, d\kappa_u.
    \end{split}
\end{equation} The replacement of the $\sigma$-algebra $\FF_t$ by $\sigma(\kappa_t, R_t)$ is permitted by
the Markov property of the process $(\kappa_t, R_t)$.

When $\kappa_t\geq 1$, the value of $G_t$ is
trivially given by (\ref{equ:prog}), so we can restrict our attention to the value of
the function $g(t,r,k)=
\EE[\int_{T_0(t)}^{\tau_1} e^{-\beta u}\, d\kappa_u | \kappa_t=k, R_t=r]$ for $k<1$, because then (\ref{equ:gone}) implies
that $G_t=g(t,R_t,\kappa_t)+\int_0^t \exp(-\beta u)\, d\kappa_u$ on $\set{\kappa_t<1}$.  Using again the strong
Markov property and time-homogeneity of $(\kappa_t, R_t)$ we obtain
\begin{equation}
    %\nonumber
    \label{equ:val}
    \begin{split}
     g(t,r,k)&
     = \EE[ e^{-\beta T_0(t)} \int_{T_0}^{\tau_1} e^{-\beta (u-T_0(t))}\, d\kappa_u\Big| R_t=r, \kappa_t=k]
    \\ &= e^{-\beta t} \EE[ e^{-\beta T_0(0)}\Big |R_0=r] \,\EE[ \int_0^{ \tau_{1-k}} e^{-\beta t} d\kappa_t
    \Big| R_0=0, \kappa_0=0].
    \end{split}
\end{equation}
The second term in the above expression is given in (\ref{equ:hit}). As for the third term,
a change of variables yields
\begin{equation}
    %\nonumber
    \label{equ:cv}
    \begin{split}
     \EE[ \int_0^{ \tau_{1-k}} e^{-\beta t} d\kappa_t]= \int_0^{1-k} \EE[ e^{-\beta \tau_u}]\, du=
     \frac{1-e^{-(1-k)\psi(\beta)}}{\psi(\beta)}
    \end{split}
\end{equation}
\end{proof}
We have developed all the tools required to prove the following result
\begin{prop} In the setup of Theorem \ref{thm:ito}, set
$U(\omega, t, x)=\exp(-\beta t) \log(x)$. Then we have the following
explicit representations of the processes $\prfi{\Hx_t}$, $\prfi{\nux_t}$
and $\prfi{\hcx_t}$:
\begin{eqnarray}
\nux_t&=&-\sgn(R_t) h\Big(\frac{\abs{R_t}}{\sqrt{2}}\Big)\ \text{
where}\ h(z)\triangleq -\frac{2\beta}{\alpha} \frac{
H_{-\frac{\beta}{\alpha}-1}( z)}{
H_{-\frac{\beta}{\alpha}}( z)},\label{equ:n1} \\
\Hx_t &=& \frac{X_t}{\sigma S_t} \Big(\theta+\rho \sgn(R_t) h(\abs{R_t}/\sqrt{2})\Big),\\
\hcx_t &=& X_t  \frac{1-\exp(-\Psi(\beta))}{(1-\exp(-(1-\kappa_t)\Psi(\beta)))}.
\label{equ:ccc}
\end{eqnarray}
Finally, the process $\prf{\nux_t}$ is bounded and so the optimal dual process $\prf{Z^{\nux}_t}$ is a
martingale.
\end{prop}
\begin{proof}
 A use of the It\^ o-Tanaka formula and the expression
(\ref{equ:prog}) yields
\begin{equation}
    %\nonumber
    \label{equ:forpsi}
    \begin{split}
     \psi^{B}_t=0,\ \text{and}\ \psi^{W}_t=
     \exp(-\beta t) \sgn(R_t) \frac{\partial}{\partial r} j(\beta, \abs{R_t})
     \frac{1-\exp(-(1-\kappa_t) \Psi(\beta))}{\Psi(\beta)}.
    \end{split}
\end{equation}
Moreover, the martingale property of process $M_t$ from
(\ref{equ:em}) implies that $X_tZ^{\nux}_t=G_t-\int_0^t e^{-\beta
u} \, d\kappa_u$, and so, equations (\ref{equ:hit}),
(\ref{equ:pi}) and (\ref{equ:psis}) can be combined into the following
explicit expression of the optimal dual process
\[ \nu^{y}_t=\sgn(R_t) \frac{\frac{\partial}{\partial \beta}
j(\beta, \abs{R_t})}{j(\beta, \abs{R_t})}.\] The representation
(\ref{equ:hit}) and the identity
$\frac{\partial}{\partial x}H_{\xi}(x)=2\xi H_{\xi-1}(x)$ (see
\citet{Leb72}, equation 10.5.2, page 289) complete the proof of (\ref{equ:n1}).

Part (7) of Theorem \ref{thm:aer}, and the identities (\ref{equ:pi}) and (\ref{equ:forpsi})
imply that
\begin{equation}
    \nonumber
    %\label{equ:none}
    \begin{split}
     \hcx_t=\frac{X_t \Psi(\beta) }{y j(\beta, \abs{R_t}) (1-\exp(-(1-\kappa_t)\Psi(\beta)))},
    \end{split}
\end{equation}
where $y$ satisfies $x=-v'(y)$. To get a more explicit expression for $y$, we combine
(\ref{equ:ppet}) and (\ref{equ:psis}) to get $xy=\EE[\int_0^{\tau_1} \exp(-\beta t)\, d\kappa_t]$.
After repeating the calculation in (\ref{equ:cv}) with $k=0$, we only need to rearrange the terms and
remember that $R_t=0$ $d\kappa$-a.e, to
obtain (\ref{equ:ccc}).

We are left with the proof of the boundedness of the process $\prfi{\nux_t}$.
The asymptotic formula 10.6.3 in \citet{Leb72}, page 291, implies
that, $H_{\xi}(x)\sim C_{\xi} x^{\xi}$ as $x\to\infty$, for some positive constant
$C_{\xi}$ depending on  $\xi<0$.
Therefore, there exists a constant $D>0$ such that $h(x)\sim D x^{-1}$, as $x\to\infty$. Because of the
existence of the limit $\lim_{x\to 0+} h(x)$, we conclude that $h$
is a bounded function on $[0,\infty)$. Hence, $\prfi{\nux_t}$ is a
bounded process, making $\prf{Z^{\nux}_t}$  a
martingale.
\end{proof}
\begin{rem} \label{rem:asslog} In the generic setup of Theorem \ref{thm:ito}, we have explicitly assumed that $u(x)<\infty$,
for at least one $x>0$. In the case of the logarithmic utility random field treated above, the validity of
such an assumption is implied by the following chain of inequalities in which
$\QQ_0$ and $Z^0_{\tau_1}$ are as in (\ref{equ:cands}).
\begin{equation}
    %\nonumber
    \label{equ:finall}
    \begin{split}
     u(x)-x & =\sup_{c\in\AA(x,0)} (\fU(c)-x) \leq \fV(\QQ_0)
     =\EE\int_0^{\tau_1} (-1-\log(Z^0_t))\, d\kappa_t\\ &\leq
     \EE[\int_0^{\tau_1} \frac{1}{2}(\theta B_t^2+1+\theta^2 t)\, d\kappa_t]
     = \frac{1}{2}\int_0^1 \EE[\theta (1+B_{\tau_s}^2)+\theta^2 \tau_s]\, ds\\
     &\leq \frac{\theta}{2}+\frac{(\theta^2+1)}{2} \int_0^1 \EE[\tau_s]\, ds
     \leq \frac{\theta+(\theta^2+1)\EE[\tau_1]}{2} <\infty.
    \end{split}
\end{equation}
The fact that $\EE[\tau_1]<\infty$ (which can easily be deduced
from (\ref{equ:laplace1})) implies both the final inequality in
(\ref{equ:finall}) and the equality
$\EE[B^2_{\tau_1}]=\EE[\tau_1]$ through Wald's identity (see
Problem 2.12, page 141 in \cite{KarShr91}).
\end{rem}

\appendix
\section{A Convex-Duality Proof of Theorem \ref{thm:aer}} We have
divided the proof into several  steps, each of which is stated
as a separate lemma. Throughout this section all the conditions of Theorem
\ref{thm:aer} are assumed to be satisfied.

\begin{lem}[Global properties of the value functions]
\label{lem:minimax} The value function $u(\cdot)$ is convex,
non-decreasing and $[-\infty,\infty)$-valued, while $v$ is
concave, and $(-\infty,\infty]$-valued. Moreover, the primal and the dual value
functions $u(\cdot)$ and $v(\cdot)$ are convex conjugates of each other.
\end{lem}
\begin{proof}\ \
\begin{enumerate}
\item Concavity of $u(\cdot)$ and convexity of $v(\cdot)$ are
inherited from the properties of the objective functions
    $\fU(\cdot)$ and $\fV(\cdot)$ (see  \citet{EkeTem99}, the proof of Lemma 2.1, p. 50, for the standard argument).
    The increase of $u(\cdot)$ follows from the
    inclusion $\AA(x,\EN)\subseteq \AA(x',\EN)$, for $x<x'$.
\item    By the Assumption \ref{ass:finite}, there exists $\tilde{x}\in\R$
such that
    $u(\tilde{x})<\infty$. It follows immediately, by concavity of
    $u(\cdot)$ that $u(x)<\infty$ for all $x\in\R$.

\item    To establish the claim that $v(\cdot)$ is the convex
conjugate of $u(\cdot)$,  we define the auxiliary domain
$\AA'(x,\EN)\triangleq \AA(x,\EN)\setminus\cup_{x'<x} \AA(x',\EN)$. Note
that
\begin{enumerate}
\item the monotonicity of the utility functional $\fU(\cdot)$
implies that \[\sup_{c\in\AA(x,\EN)} \fU(c)=\sup_{c\in\AA'(x,\EN)}
\fU(c),\ \text{and}\]
\item the Proposition \ref{pro:charadm}
implies that $\sup_{\QQ\in\DDk(y)} \scl{c-e}{\QQ}=xy$, for any
$y>0$, and $c\in \AA'(x,\EN)$.
\end{enumerate}
Having established the weak-* compactness of the dual domain
$\DDk(y)$ in \ref{pro:iscomp}, the Minimax Theorem (see \citet{Sio58}) implies that
\begin{equation}
\nonumber \label{equ:equ1}
\begin{split}
\sup_{x\in \R} [u(x)-xy] &=\sup_{x\in \R} \Big(
\sup_{c\in\AA'(x,\EN)} \fU(c) - xy \Big)
\\ &= \sup_{x\in\R}\  \sup_{c\in \AA'(x,\EN)} \Big( \fU(c) - \sup_{\QQ\in\DDk(y)}\scl{c-e}{\QQ} \Big)  \\
&= \sup_{x\in\R}\  \sup_{c\in \AA'(x,\EN)} \inf_{\QQ\in\DDk(y)}
\Big( \fU(c) - \scl{c}{\QQ}+\scl{e}{\QQ} \Big)  \\
&= \sup_{c\in\slzmkp}\ \inf_{\QQ\in\DDk(y)} \Big( \fU(c) -
\scl{c}{\QQ}+\scl{e}{\QQ}\Big) \\ &=\inf_{\QQ\in\DDk(y)}
\sup_{c\in\slzmkp}\Big( \fU(c) - \scl{c}{\QQ}+\scl{e}{\QQ}\Big)\\ & =
\inf_{\QQ\in\DDk(y)} \Big( \fV(\QQ)+\scl{e}{\QQ} \Big) = v(y).
\end{split}
\end{equation}

\end{enumerate}
\end{proof}

\begin{lem}[Existence in the dual problem]
For \label{lem:dualexists} $y\in\Dom(v)$ there exists
$\hqy\in\DDk(y)$ such that
\[ v(y)=\fVE(\hqy)=\fV(\hqy)+\scl{e}{\hqy}.\]
\end{lem}
\begin{proof}
For $y\in\Dom(v)$, let $\sq{\QQ_n}$ be a minimizing sequence for
$v(y)$, i.e. a sequence in $\DDk(y)$, such that $\sq{\fVE(\QQ_n)}$
is real-valued and decreasing with limit $v(y)$. Since $\DDk(y)$ is a closed and bounded subset of the dual
$(\slzmk)^*$ of $\slzmk$. By Proposition \ref{pro:iscomp} the product space $\DDk(y)\times [v(y),
\fVE(\QQ_1)]$ is compact. Therefore
  the sequence
$\sqb{\QQ_n, \fVE(\QQ_n)}$ has a cluster point $(\hqy, v^*)$ in
$\DDk(y)\times [v(y), \fVE(\QQ_1)]$. By the decrease of the
sequence $\sq{ \fVE(\QQ_n)}$, we have  $v^*=\lim_n
\fVE(\QQ_n)=v(y)$. On the other hand, by the definition
(\ref{equ:defV}) of the functional $\fV(\cdot)$ , the epigraph of
its restriction $\fVE(\cdot):\DDk(y)\to\R$ is closed with respect
to the product of the weak-* and Euclidean topologies. Therefore,
$(\hqy, v^*)$ is in the epigraph of $\fVE$ and thus, $v(y)=v^*\geq
\fVE(\hqy)=\fV(\hqy)+\scl{\hqy}{e}.$
\end{proof}

\begin{lem}[Consequences of Reasonable Elasticity]\label{lem:difv}\
\begin{enumerate}
\item $\Dom(v)=(0,\infty)$.
\item $v(\cdot)$ is continuously differentiable, and for $y>0$ its
derivative satisfies
\[ yv'(y)=-\scl{(\hqy)^r}{\fI(\hqy)}+\scl{e}{\hqy},\] where
$\hqy\in\DDk(y)$ is a minimizer in the dual problem, i.e.
$v(y)=\fVE(\hqy)$.
\item The following inequality holds for
all $\QQ\in\DDk(y)$
\[ yv'(y)\geq -\scl{\QQ^r}{\fI(\hqy)}+\scl{e}{\hqy}.\]
\item $\lim_{y\to 0} v'(y)=-\infty$ and $\lim_{y\to\infty}
v'(y)\in \big[\inf_{\QQ\in\MM}\EE^{\QQ}[\EN_T],\sup_{\QQ\in\MM}\EE^{\QQ}[\EN_T]\big]$
\item $\fI(\hqy)\in \AA(-v'(y), e)$ and $\scl{\fI(\hqy)}{(\hqy)^r}=\scl{\fI(\hqy)}{\hqy}$.
\end{enumerate}
\end{lem}
\begin{proof} Thanks to the representation
$ v(y)=\EE\int_0^T V(t,Y^{\hqy}_t)\, d\kappa_t$, and the fact that
$\EE\int_0^T \yq_t\, d\kappa_t\leq 1$ for all $\QQ\in\DDk(1)$, the
proofs of parts (1)-(4) this lemma follow (in an almost verbatim
fashion)  the proofs of the following statements in
\citet{KarZit03}: (1) Lemma A.5, p.30, (2) Lemma A.6, p. 31., (3)
Proposition A.7, p. 32., and (4)  Lemma A.8, p. 33.

To prove the claim (5), we observe that the combination of (3) and (4) implies
that
\[ \scl{\fI(\hqy)}{y\QQ}\leq -yv'(y)+\scl{e}{y\QQ},\ \text{for all $\QQ\in\MMk$.}\] From Proposition
\ref{pro:charadm2} it follows that $\fI(\hqy)\in \AA(-v'(y),e)$, so
$\scl{\fI(\hqy)}{\QQ}\leq -yv'(y)+\scl{e}{\QQ}$, for all $\QQ\in\DD(y)$. In particular,
$\scl{\fI(\hqy)}{\hqy}\leq -yv'(y)+\scl{e}{\hqy}$, yielding immediately the inequality
$\scl{\fI(\hqy)}{\hqy}\leq \scl{\fI(\hqy)}{(\hqy)^r}$. The second part of the claim
follows by the trivial inequality $\scl{\fI(\hqy)}{\hqy}\geq \scl{\fI(\hqy)}{(\hqy)^r}$.
\end{proof}

\begin{lem}[Existence in the Primal Problem]\label{lem:exist}
For $x>-\lim_{y\to\infty} v'(y)$ the Primal Problem ($\ref{equ:problem}$) has a solution,
i.e. there exists $\hcx\in\AA(x,\EN)$
such that $u(x)=\fU(\hcx)$. Moreover, the optimal consumption density process $\hcx$ is $\PPk$-a.s. unique.
\end{lem}
\begin{proof}
Using the continuous differentiability of the dual value function $v(\cdot)$ and Lemma \ref{lem:expropu},
we conclude that for any $x>\lim_{y\to\infty} v'(y)$ there exists a unique $y>0$ such that $v'(y)=-x$.
Let $\hqy$ be the solution to the dual problem corresponding to $y$, and define the candidate solution $\hcx$
to the primal problem by
\[ \hcx\triangleq \fI(\hqy).\] By Lemma \ref{lem:difv} $\hcx\in\AA(x,\EN)$.
The optimality of the consumption density process $\hcx$ follows from the fact that
\begin{equation}
    \nonumber
    %\label{equ:none}
    \begin{split}
\fU(\hcx)&=\fU(\fI(\hqy))=\fV(\hqy)+\scl{\fI(\hqy)}{\hqy}=
\fV(\hqy)+\scl{\fI(\hqy)}{(\hqy)^r}\\ &=v(y)-y v'(y)=u(x),
    \end{split}
\end{equation}
 using Lemma \ref{lem:difv} and
the conjugacy of $u(\cdot)$ and $v(\cdot)$.
The $\PPk$-a.s. uniqueness of $\hcx$ is a direct consequence of the strict concavity of the mapping
$x\mapsto U(\omega, t, x)$ coupled with convexity of the feasible set $\AA(x,\EN)$.
\end{proof}

\begin{lem}\label{lem:expropu}\
$\lim_{y\to\infty} v'(y)=\LL(\EN)$, where $\LL(\EN)=\inf_{\QQ\in\MM} \EE^{\QQ}[\EN_T]$.
\end{lem}
\begin{proof}
Let $x'=\lim_{y\to\infty} v'(y)$. Part (4) of Lemma \ref{lem:difv} states that $x'\geq \LL(\EN)$, so
we only need to prove that $x'\leq \LL(\EN)$. Suppose, to the contrary, that there exists $x_0>\LL(\EN_T)$
of the form $x_0=v'(y_0)$ for some $y_0>0$ so that $x'> x_0$.
The optimal consumption process $\prf{C^{-x_0}_t}$ corresponding to the initial capital $-x_0$ exists by the
Lemma \ref{lem:exist} and satisfies $\EE^{\QQ}[C^{-x_0}_T]\leq -x_0+\EE^{\QQ}[\EN_T]$ for any $\QQ\in\MM$
by Proposition \ref{pro:charadm}. Taking the infimum over $\QQ\in\MM$ we reach a contradiction
\[ 0\leq \inf_{\QQ\in\MM} \EE^{\QQ}[C^{-x_0}_T]\leq -x_0+\LL(\EN_T)<0.\] Therefore, $x' \leq \LL(\EN)$.
\end{proof}

\def\cprime{$'$}

\end{document}